\documentclass[prx,aps,superscriptaddress,footinbib,twocolumn]{revtex4-2}
\usepackage{amsmath}
\usepackage{mathtools}

\newcommand{\abs}[1]{\vert #1 \vert}

\newcommand{\Tr}{\operatorname{Tr}}
\newcommand{\norm}[1]{\Vert #1 \Vert}

\newcommand{\ket}[1]{\vert{ #1 }\rangle}
\newcommand{\bra}[1]{\langle{ #1 }\vert}

\usepackage{amsthm}
\newtheorem{theorem}{Theorem}

\newtheorem{lemma}{Lemma}
\newtheorem{corollary}{Corollary}

\theoremstyle{definition}
\newtheorem{definition}{Definition}
\theoremstyle{remark}

\usepackage{amssymb}

\usepackage{graphicx}
\usepackage{ragged2e}
\usepackage{subcaption}

\makeatletter
\newcommand*\bigcdot{\mathpalette\bigcdot@{.5}}
\newcommand*\bigcdot@[2]{\mathbin{\vcenter{\hbox{\scalebox{#2}{$\m@th#1\bullet$}}}}}
\makeatother

\usepackage{natbib}

\usepackage[colorlinks=true,citecolor=blue,linkcolor=magenta]{hyperref}

\usepackage{xcolor}

\definecolor{cg}{rgb}{0,0.6,0}
\definecolor{co}{rgb}{128, 0, 128}
\begin{document}
	
	\title{Super-bath Quantum Eigensolver}
	
	\author{Tianren Wang}
	\affiliation{Graduate School of China Academy of Engineering Physics, Beijing 100193, China}
	
	\author{Zongkang Zhang}
	\affiliation{Graduate School of China Academy of Engineering Physics, Beijing 100193, China}
	
	\author{Bing-Nan Lu}
	\affiliation{Graduate School of China Academy of Engineering Physics, Beijing 100193, China}
	
	\author{Mauro Cirio}
	\email{cirio.mauro@gscaep.ac.cn}
	\affiliation{Graduate School of China Academy of Engineering Physics, Beijing 100193, China}
	
	\author{Ying Li}
	\email{yli@gscaep.ac.cn}
	\affiliation{Graduate School of China Academy of Engineering Physics, Beijing 100193, China}

	\begin{abstract}
		 The simulation of the dynamics of a system coupled to a low-temperature environment is a promising  application of quantum computers to determine  ground-state  properties of physical systems. However, this approach requires not only the \textit{existence} of an environment that allows the system to dissipate energy and evolve to its ground state, but also the \textit{detailed knowledge} of the properties of the bath. In this paper, we propose a polynomial-time algorithm for ground state preparation which only relies on the \textit{existence} of a physical bath which achieves the same task, while a detailed description of the environment may remain \textit{unknown}. In particular, we show that this  ``super-bath quantum eigensolver algorithm'' prepares the ground state of the system by combining a Gaussian stabilization dephasing procedure with the simulation of the interaction between the system and a super-bath which only requires minimal knowledge of the physical environment. Based on our algorithmic framework, we establish a partial order relation among environments. Supported by experimental lifetime data of nuclear metastable states, we suggest that our algorithm is applicable to determine nuclear ground states in polynomial time. These results highlight the potential advantage of quantum computing in addressing ground state problems in real-world physical systems.	
	\end{abstract}
	
	\maketitle

	\section{Introduction}
	 The possibility to use quantum computers to find the ground state of physical systems is one of their most promising applications because of the potentially broad impact on fields  such as physics, chemistry, and materials science \cite{cao2019quantum,oftelie2021simulating,camino2023quantum,mazzola2024quantum}.
	However, the practical question of whether quantum algorithms can solve for these ground states in polynomial time is currently still an open question \cite{dong2022ground,shang2024polynomial}. The core challenge lies in preparing a quantum state having a sufficient overlap with the ground state. Once such a state is available, quantum phase estimation (QPE) \cite{kitaev1995quantum, dorner2009optimal, nielsen2010quantum, Wan2021A, clinton2024quantum} or other projection algorithms \cite{motta2020determining, huggins2022unbiasing,epperly2022theory,xu2023quantum,Zhang2024measurement} can  be applied to amplify the initial overlap. For this reason, research efforts have been focusing on initializing quantum systems, resulting in different methods such as ansatz state preparation (e.g. Hartree-Fock and Kohn-Sham ground states) \cite{PhysRevX.6.031007,Lee2023}, adiabatic state preparation \cite{Lee2023,whitfield2011simulation}, variational quantum eigensolvers \cite{peruzzo2014variational,PhysRevX.6.031007,google2020hartree,endo2020variational},	
	and the simulation of dissipative dynamics \cite{lloyd,Verstraete2009,wang2011quantum,di2015quantum,	chenu2017quantum,su2020quantum,hu2020quantum,schlimgen2021quantum,cattaneo2021collision,de2021quantum,kamakari2022digital,de2022quantum,suri2023two,wang2023simulating,Delgado-Granados2025}, which has  recently been receiving increased attention \cite{Polla2021,Metcalf2022,Chen2021fast,Layden2023,chen2023efficient,rall2023thermal,Cattaneo2023,movassagh2023preparing,Rost2021,PhysRevResearch.6.043229,Jiang2024,Gilyen2024,Chen2024,Ding2023,rouzé2024optimalquantumalgorithmgibbs,lloyd2024quasiparticle}. While some of these techniques have been numerically shown to work efficiently for small system sizes, their scaling is still an open question \cite{Lee2023}.  As a consequence, their application to achieve a quantum  advantage in solving physical ground states remains uncertain  \cite{Verstraete2009,xu2014demon,chen2023quantum,Cubitt2023}.
	
     Despite the challenges of solving this problem in full generality, it is interesting to note that, in nature, many systems, such as nuclei and molecules, do exist in their ground states \cite{PhysRevLett.131.212501}. These systems typically interact with a low-temperature thermal bath which  brings them to the ground state by allowing a steady dissipation of energy \cite{kuzemsky2022exotic,chen2024local}. A slow decrease of the dissipation rate as a function of the system size is crucial for the observation of ground states in large systems \cite{wen2018two,kim2021theory,harrington2022engineered}. In fact, it is reasonable to assume that systems which are naturally found in the ground state are in contact with a bath whose energy dissipation power  scales polynomially with the system size. For example, we will show that nuclei support such a scaling.

	 In general, one would expect that, whenever such an efficient dissipation exists in nature,  a quantum digital simulation could be achieved with a low time complexity. However, this approach implicitely assumes an accurate model of the bath and the system-bath interaction which would require experimental data and a theoretical analysis. As a consequence, the success of a quantum simulation of a natural bath will critically depend on the accuracy of this procedure, particularly when the energy dissipation rate is sensitive to specific details of the bath model. On the other hand, it might be possible to construct bath models that lead to efficient dissipation without necessarily trying to emulate nature. This could be done by optimizing over a large space of candidate models \cite{schirmer2010stabilizing}. To  reduce the reliance on accurate prior information or on extensive optimizations, here we propose a ``super-bath'' quantum eigensolver algorithm for preparing ground states, which is robust against variations in the bath model.  This is achieved by establishing a partial order relation among environments within our algorithmic framework.

	Suppose a low-temperature bath with sufficiently rich channels to induce transitions between a system's energy levels such that it can drive the system to its ground state in polynomial time. Intuitively, a more complex bath is expected to retain this capability. However, this intuition is flawed, as it holds only when the system's evolution can be described by the Lindblad equation, which often fails due to the breakdown of the rotating wave approximation in many-body systems. Our algorithm advances dissipation-based algorithms for ground state preparation by iteratively simulating an open system dynamics and a randomised time evolution. We demonstrate that short-time evolution can be restored to the Lindblad picture, which is sufficient to empower a super-bath containing a sub-bath with the same dissipation capability in the long-time evolution, up to errors which can be controlled by tuning the algorithm parameters.

	 This paper is structured as follows. In section \ref{sec:notations}, we start by reviewing the general formalism of open quantum systems. In section \ref{sec:gamma decay}, we introduce the gamma decay of nuclei as an example of a physical system evolving towards its ground state thanks to the interaction with its environment. We further provide evidence that the corresponding dissipation rates scale favorably with the system size. In section \ref{sec:superbath eigensolver}, we present the different sub-modules constituting our algorithm. We motivate this structure by showing, in section \ref{sec: robustness}, the main physical intuition behind its capability to  drive the system towards the ground state and follow with a rigorous analysis of the algorithm's performance and efficiency. In section \ref{sec:implementation and scalability}, we raise potential issues about implementation and scalability and finish with our overall conclusions in section \ref{sec:conclusions}. All explicit details about the technical steps of the derivations can be found in the Appendix. 
	
	\section{Notations of open quantum systems}
	\label{sec:notations}

	In this section, we introduce the notation for the description of open quantum systems which we are going to use throughout the text. We refer to Appendix~\ref{app:open_system} for a more in-depth overview of these concepts. 

	In general, an open quantum system consists of a quantum system interacting with an external continuum, such as a thermal bath. The open system can be described by $(H,H_B,H_I,T)$, where $H$ is the Hamiltonian of the system, $H_B$ is the Hamiltonian of the bath, $H_I$ is the interaction Hamiltonian, and $T$ is the temperature of the bath.  Here, we focus on an interaction Hamiltonian which can be expressed in the form $H_I = \sum_{\alpha = 1}^N A_\alpha\otimes B_\alpha$, where $A_\alpha$ and $B_\alpha$ are system and bath Hermitian operators, respectively.  We  define  $\mathbb{A} = \{A_\alpha\vert\alpha = 1,2,\ldots,N\}$ to denote the set of these system operators, which, in general can  represent many-body interactions. For example, for a system of qubits, the elements in $A_\alpha$ can be chosen as a basis for 1-qubit operators while,  in the Fermionic case, one-Fermion operators. Given a choice for the set $\mathbb{A}$, the open system can be described by $(H,H_B,\{B_\alpha\},T)$. As two specific examples, the set $\mathbb{A}$ can be chosen as follows.

\begin{definition}
	{\bf Basis operators of qubit systems.} For a system of $n$ qubits, the basis operator set is 
	\begin{eqnarray}
		\mathbb{A} = \{X_j,Y_j,Z_j \vert j=1,2,\ldots,n\}
	\end{eqnarray}
	where $X_j,Y_j,Z_j$ are Pauli operators acting on the $j$-th qubit. 
	\label{def:BOQS}
\end{definition}

\begin{definition}
	{\bf Basis operators of fermion systems.} For a system of $n$ fermion modes, the basis operator set is 
	\begin{eqnarray}
		\mathbb{A} &=& \{c_j^\dag c_j,c_i^\dag c_j+c_j^\dag c_i,i(c_i^\dag c_j-c_j^\dag c_i) \vert \notag \\
		&& i,j=1,2,\ldots,n,i\neq j\}
	\end{eqnarray}
	where $c_j,c_j^\dag$ are annihilation and creation operators, respectively, of the $j$-th mode. 
	\label{def:BOFS}
\end{definition}
For a Gaussian Bosonic bath, all its effects on the system can be fully	encoded into a matrix-valued spectral density $J_{\alpha,\beta}(\omega)$ which is semi-positive definite for all the frequencies $\omega\geq 0$, and where $\alpha,\beta=1,\dots N$. As a consequence, For a Gaussian Bosonic bath,  the open quantum system can be described by $(H,J,T)$. In fact, the system dynamics up to the time $s$ is determined by  the bath correlation function $C_{\alpha,\beta}(T,J;s)$. It is possible to further define $\gamma_{\alpha,\beta}(T,J;\omega)$ and $\Gamma_{\alpha,\beta}(T,J;\omega)$ as  the Fourier and half-Fourier transforms of $C_{\alpha,\beta}(T,J;s)$, respectively. The correlation functions also determine both the typical time-scale $\tau_R(T,J)$ for the system relaxation and the time-scale $\tau_B(T,J)$ for the bath correlation. Because of the influence of the bath, the system dynamics follows a time-dependent trace-preserving completely positive map $\mathcal{M}$. For a Gaussian Bosonic bath, this map only depends on the system Hamiltonian $H$, the temperature $T$, the spectral density $J$ and the evolution time $t$, i.e. $\rho(t) = \mathcal{M}(H,T,J;t)\rho(0)$, where $\rho(t)$ is the state of the system at the time $t\geq 0$. 

A relevant regime for open quantum systems is when the system is weakly coupled to a large external bath so that the influence of the system on the bath is negligible while the timescale over which the state of the system changes appreciably is large compared to the timescale over which the bath correlation functions decay \cite{breuer2002theory}.  In this case, the so-called Born-Markov approximation is valid and the resulting system dynamics can be described by a Markovian master equation, referred to as the Redfield equation. This equation reads $\dot{\rho}(t) = -i[H,\rho(t)] + \mathcal{K}(T,J;0) \rho(t)$,  in terms of the superoperator $\mathcal{K}(T,J;t)$ describing the effect of the interaction with the bath. For Gaussian Bosonic baths, there is a rigorous upper bound on the error made by imposing the Born-Markov approximation. This error approaches zero in the limit $\tau_B(T,J)/\tau_R(T,J)\rightarrow 0$ \cite{nathan2020universal}. 

Under the condition that the system relaxation rate $1/\tau_R$ is small compared to  the transition energies of $H$,  it is possible to apply rotating-wave approximation to the Redfield equation  which further simplifies to become the Lindblad equation  \cite{breuer2002theory}. The Lindblad equation can be written as $\dot{\rho}(t) = -i[H+H_{LS}(T,J),\rho(t)] + \mathcal{L}(T,J) \rho(t)$,  in terms of  the Lamb shift Hamiltonian $H_{LS}(T,J)$ which satisfy $[H,H_{LS}] = 0$ and the dissipator $\mathcal{L}(T,J)$. This concludes the introduction of open quantum systems. In the next section, we are going to present a specific example where the algorithm we present could be applied. 
	
\section{Gamma decay in nuclei}
	\label{sec:gamma decay}
	
	\begin{figure*}[t]
		\centering
		\includegraphics[width=\linewidth]{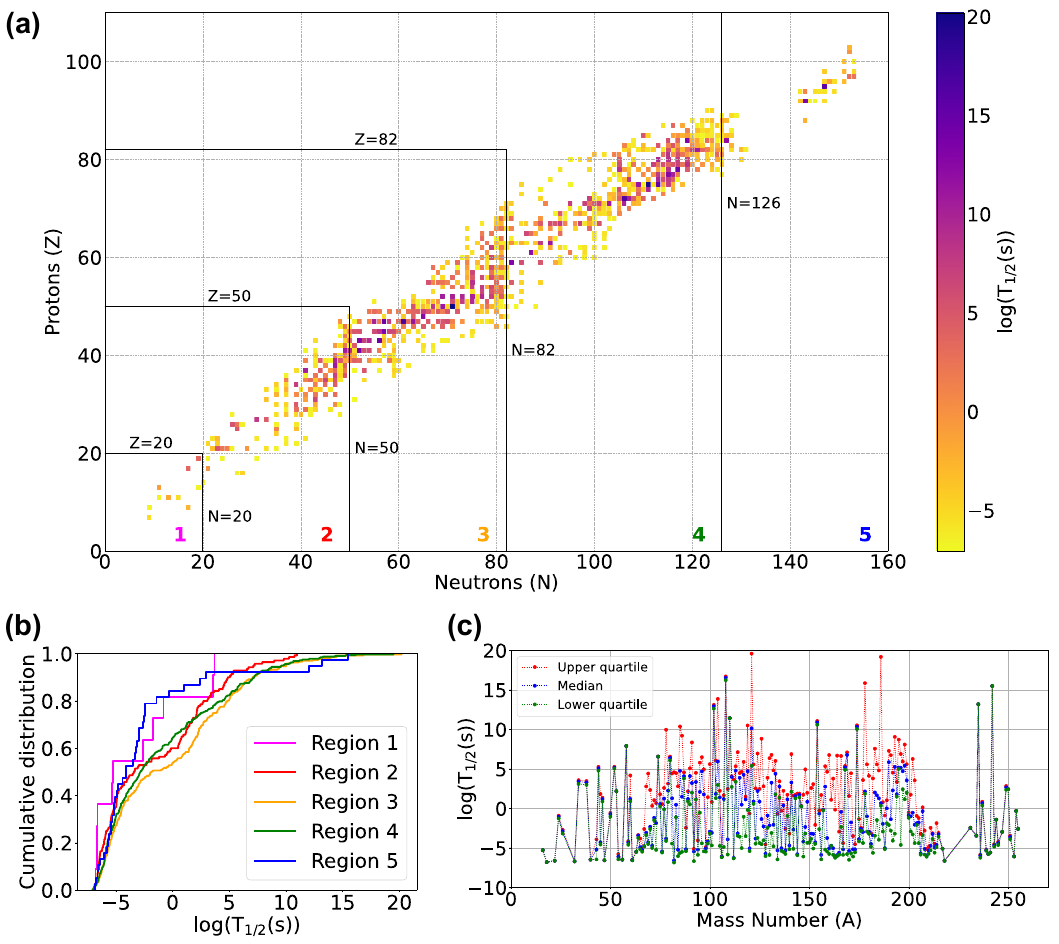}
		\captionsetup{justification=Justified,singlelinecheck=false,font=small}
		\caption{The gamma decay half-lives $\textrm{T}_{1/2}$ of nuclear isomers. If a nuclide has multiple isomers, only the longest-lived state is included. (a) A nuclear chart depicting the half-lives of isomers. Based on the magic quantum numbers of nuclei, namely $20$, $50$, $82$, and $126$, isomers are divided into $5$ regions. (b) The cumulative distribution of half-lives within each of these 5 regions. (c) The quartiles of half-lives for nuclides with the same mass number. The data used to generate this figure is sourced from the NUBASE2020~ \cite{Kondev_2021}, the NuDat~\cite{nudat}, and the conversion coefficient calculator BrIcc~\cite{KIBEDI2008202}. }
		\label{fig:isomer}
	\end{figure*}

	The nuclear excited states can decay via all channels allowed by symmetries, for example, the excitation energies can be carried off by gamma rays, nucleons, light clusters such as $^4$He or surrounding electrons. Heavy nuclei can also fission where the excitation energies are converted into kinetic energies of the fragments. The branching ratio is determined by the structure of both the excited and the ground states. The specific decay mechanism is usually complicated and varies case-by-case. A universal description is still not available and poses a great challenge to the nuclear theorists. Among the various decay mechanisms, the gamma decay, where the nucleus interacts with the electromagnetic field and emits photons carrying specific energy, momentum and angular momentum, is the mostly studied one. The emitted photons can carry important information of the nuclear structure and the gamma-ray spectroscopy has been dubbed the fingerprint of the nucleus. To avoid unnecessary complexity, here we only consider simple gamma decay with the photon energy $E_\gamma\gtrsim 1 \mathrm{keV}$ and half-life $\textrm{T}_{1/2}$. The energy dissipation power of the excited state is defined as 
	\begin{eqnarray}
		P = \frac{E_\gamma\ln 2}{\textrm{T}_{1/2}}. \label{eq:EDP}
	\end{eqnarray}
	Using experimental gamma-decay half-lives, we first investigate how the half-life (thereby, energy dissipation power) varies against the nuclear mass number. 
	
	In Fig.~\ref{fig:isomer} we summarize the gamma decay lifetime throughout the whole chart of nuclide. Generally, the lifetime of the excited states is of the order of $10^{-12}$ seconds. However, some metastable states, called isomers \cite{krane1991introductory}, have significantly longer lifetimes than typical excited states. In the figure, we show the longest isomer half-life of each nucleus. One observation from the data is that the order of gamma-decay half-lives is almost independent of the system size, suggesting favorable scaling of the energy dissipation power. However, for specific nucleus, the longest half-life can vary by a few tenths of orders from $10^{-7}$ to $10^{20}$ seconds. We note that for those nuclei containing isomers with extremely long lifetimes, the natural bath method may be stuck in that state for a long time and becomes inefficient for preparing the corresponding ground state. This issue necessitates the search for a { synthetic} bath that induces rapid dissipation or an algorithm that does not rely on the details of the bath.  
	
		\begin{figure*}[t]
		\centering
		\includegraphics[width=\textwidth]{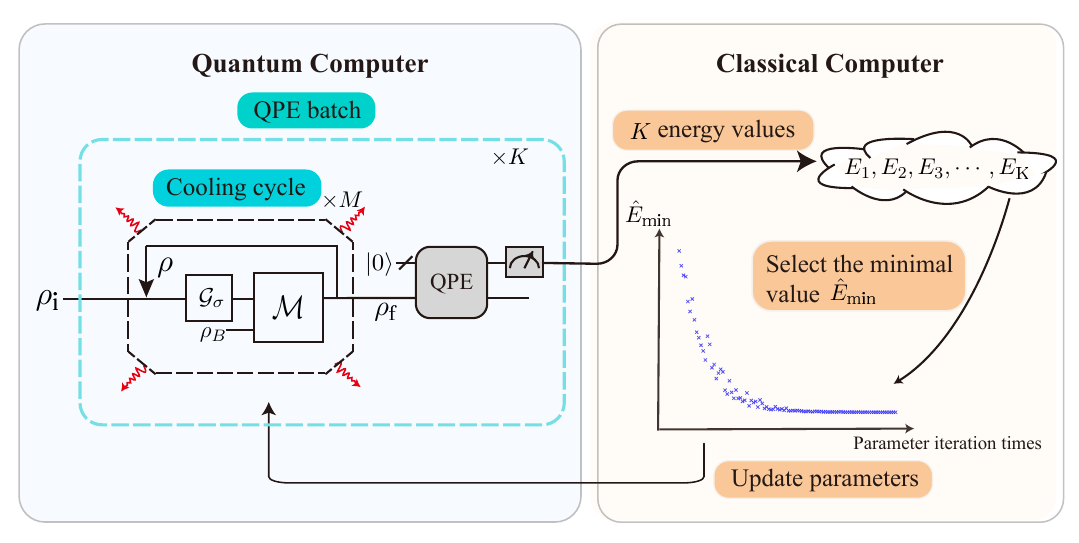} 
		\captionsetup{justification=Justified, singlelinecheck=false, font=small}
		\caption{Schematic diagram of the super-bath quantum eigensolver. The quantum computer receives fixed parameters $(H,\mathbb{A},\mathfrak{J})$, tunable parameters $(T,g,t,\sigma,M)$, and an initial state 
		$\rho_{\text{i}}$. Through the iteration of $M$ cooling cycles --- each comprising a Gaussian stabilization and an open-system dynamics channels --- the system attains the ground state with a finite probability. Consequently a QPE protocol is used to measure the system energy. This entire procedure is repeated $K$ times, thereby generating 
		$K$ energy eigenvalues. The minimal output $\hat{E}_{\text{min}}$ is recorded, followed by an update of the tunable parameters which are fed back into the quantum computer to restart the protocol. The algorithm terminates upon convergence of the energy output, yielding the ground state of the target system.}
		\label{fig: Algorithm_test}
	\end{figure*}

\section{Super-bath quantum eigensolver}
\label{sec:superbath eigensolver}

	In this section, we provide a general overview about the structure of the proposed algorithm as shown in Fig. \ref{fig: Algorithm_test}, which we refer to as the super-bath quantum eigensolver (SQE). Given a quantum system with Hamiltonian $H$, the algorithm estimates  its ground state energy. This is achieved by evolving a trial state through a cooling cycle made out of $M$ repetitions of a dephasing channel and an open dynamics, after which the state energy is evaluated using quantum phase estimation (QPE). This last step involves the classical processing of $K$ energy values which are the output of the previous quantum processing. By selecting the smallest energy, the classical computer can estimate whether the optimal ground state energy has been reached or the parameters defining the quantum part of the algorithm are required to be updated. By repeating this process, the algorithm eventually yields the ground state energy of the system. The efficiency of the algorithm will be discussed in the next section. We are now going to describe each of the parts of the algorithm in more detail.

	 $\bullet \quad${\bf Input.} We define the input of the algorithm as the set of operators and parameters used to tune the different algorithmic components. They can be listed as $(H,\rho_{\text{i}},\mathbb{A},\mathfrak{J},T,g,t,\sigma,M)$. Here, $H$ and $\rho_{\text{i}}$ are the Hamiltonian and the initial state (which is arbitrary and can always be chosen as the maximally mixed state) for the system. The quantities $\mathbb{A}$ and $\mathfrak{J}$  are used to define the open dynamics. In particular,  $\mathbb{A}$ is the set of system interaction operators and $\mathfrak{J}$ is a proper spectral density function (see Definition \ref{def:PSDF}). This open dynamics also depends on the parameters $T$, $g$, and $t$ setting the temperature, the overall scaling for the strength of the system-environment coupling, and the evolution time.  The parameter $\sigma$ is used to tune the dephasing channel, while $M$ counts the times the dephasing and open dynamics channels are iterated within each cooling cycle. The paramers update only involves the parameters $(T,g,t,\sigma,M)$. We are now going to give more details about each of the components introduced above.

	$\bullet \quad${\bf Cooling cycle.} The cooling cycle is the most important part of the algorithm. It consists of $M$ iterations of two operations: a dephasing, Gaussian stabilization  $\mathcal{G}_\sigma$ and a simulation $\mathcal{M}$ for an open system dynamics. We are going to describe these two channels in more detail.
	
	\begin{definition}
		{\bf Gaussian stabilization.}   Given a zero-mean Gaussian probability density function  $g_\sigma(t)$ having standard deviation $\sigma$, the Gaussian stabilization dephasing operation is defined as   
		\begin{eqnarray}
			\mathcal{G}_\sigma \bullet = \int_{-\infty}^{+\infty} dt g_\sigma(t) e^{-iHt} \bullet e^{iHt},
		\end{eqnarray}
		which is a completely positive map.  
		\label{def:GS}
	\end{definition}
	
	Using the map defined above, given an arbitrary state of the system, we can generate an approximate stationary state by evolving it for a randomly chosen evolution time. This is a crucial step for error control of the algorithm. Notice that events with $\abs{t}>10\sigma$ are extremely rare (the probability is about $10^{-23}$). Therefore, the time  required for this Gaussian stabilization is effectively finite.
	
 After this dephasing operation,  we need to simulate the interaction with a Gaussian boson bath, having temperature $T$, and a spectral density given by $g^2J_{\text{S}}$. Here, the positive parameter $g$ controls the overall coupling strength, and $J_{\text{S}}$ is the ``super-bath'' spectral density, defined as follows.
	
	\begin{definition}
		{\bf Super-bath spectral density.} Given a { proper} spectral density function $\mathfrak{J}$, the super-bath spectral density is 
		\begin{eqnarray}
			J_{\text{S}}(\omega) = \mathfrak{J}(\omega)\openone_{\text{N}},
		\end{eqnarray}
		where $\openone_{\text{N}}$ is the identity operator of the dimension $N = \abs{\mathbb{A}}$. 
		\label{def:superBSD}
	\end{definition}
	A proper spectral density is defined as 
	 \begin{definition}
		{\bf Proper spectral density function.} A spectral density function $\mathfrak{J}(\omega)$ is said to be proper if and only if 
		\begin{itemize}
			\item[1.] It must be a non-negative function; 
			\item[2.] $\mathfrak{J}(0) = \mathfrak{J}(\infty) = 0$, and $C(T,\mathfrak{J};0)$ is a bounded function of $T\in(0,1]$; 
			\item[3.] The bath correlation time may increases with $1/T$ but no faster than polynomially, i.e. $\tau_B(T,\mathfrak{J}) = O(\mathrm{Poly(1/T)})$; 
			\item[4.]There are two positive numbers $\tau_{R,m}$ and $\tau_{R,M}$ such that $\tau_{R,m} \leq \tau_R(T,\mathfrak{J}) \leq \tau_{R,M}$ when $T\in(0,1]$. 
		\end{itemize}
		\label{def:PSDF}
	\end{definition}
		In the above definition, the choice of the supremum in the temperature interval (we have taken the supremum $T = 1$)  is arbitrary. 
	
	Intuitively, the super-bath spectral density describes a simplified bath in which all system coupling operators interact with an independent version of a proper, i.e., well behaved, spectral density function $\mathfrak{J}$.  The conditions defining the properties of $\mathfrak{J}$ have the following meanings.
	Condition 1 is a fundamental requirement of spectral density functions. {In fact,} physical spectral densities typically tend to zero at low and high frequency such as the Ohmic, sub-Ohmic, and super-Ohmic baths {and they are also characterizes by high-frequency cut-offs}.  { Furthermore, condition 2 imposes that the integral of the spectral density multiplied by $\coth({\omega}/{2T})$ over positive energies} (i.e., the correlation function at zero time) is finite. {Since}   we employ a specific definition of the bath correlation time to obtain rigorous bounds on various errors in the algorithms, the bath correlation time may be divergent at zero temperature. { While} this divergence is tolerable, we require {it to be slow with }condition 3. Condition 4  {further requires} that the relaxation rate of the system is finite at a finite temperature, which is always true in physical systems. An example of  {a proper} spectral density functions is the super-Ohmic spectral density with exponential cutoff (see Appendix \ref{app: example}).

	This open dynamics defines a map $\mathcal{M}(H,T,g^2J_{\text{S}};t)$ on the system. {Together with the previously introduced} dephasing operation $\mathcal{G}_\sigma$, these maps are repeated  $M$ times, defining the cooling cycle which  allows energy to be dissipated. With these notations,  the output of the cooling cycle is the system state $\rho_{\text{f}} = [\mathcal{M}(H,T,g^2J_{\text{S}};t)\mathcal{G}_\sigma]^M\rho_{\text{i}}$.

	 $\bullet \quad${\bf QPE trial and batch.} { After the cooling cycle, we apply a QPE procedure to measure the state energy.} The measurement outcome is the ground state energy with a probability of $\bra{G} \rho_{\text{f}} \ket{G}$ (up to errors in the QPE), where $\ket{G}$ is the normalised ground state of $H$. When the outcome is the ground state energy, the state $\rho_{\text{f}}$ is projected onto $\ket{G}$. Because the outcome is probabilistic, we repeat the QPE trial for $K$ times to { increase the probability to }observe the ground state. {We call the set of $K$ QPE trials a ``QPE batch''}. 
	 
	  From the K measurement outcomes of each QPE batch, we select the minimum value $\hat{E}_{\text{min}}$ which constitutes the best estimate for the ground state energy. If $\bra{G} \rho_{\text{f}} \ket{G}$ is small, we may fail to observe the ground state in the QPE batch, and $\hat{E}_{\text{min}}$ { will result} higher than the ground state energy. The failure probability is $(1-\bra{G} \rho_{\text{f}} \ket{G})^K$. If we succeed in observing the ground state in the batch, $\hat{E}_{\text{min}}$ is the ground state energy (up to errors in QPE). 
	 
	 Under the condition that $\bra{G} \rho_{\text{f}} \ket{G} \geq \frac{1}{4}$, we can observe the ground state in the batch with a probability higher than $1-\kappa$, { after} $K = O(\log(\kappa))$ {iterations}.
	
 $\bullet \quad${\bf Parameter update.} In the situation that the proper parameters are unknown, we have to find them by updating parameters following certain rules. Fortunately, we know the directions to look for proper parameters: a lower temperature $T$, weaker coupling $g$, more sophisticated dephasing (larger $\sigma$) and more rounds of the cooling cycle (a larger $M$) are always preferred. Therefore, we start with some default initial values of parameters and update them in a loop. At the beginning of the algorithm, we take default initial values of parameters $T = 1/10$, $g = 1/2$, $\sigma = 10$ and $M = 10^3$ (or any other default values). We always take $t = \tau_R(T,J_\text{S})/g$. We start with a finite temperature instead of taking the zero temperature directly because of the bath correlation time: if $\tau_B(T,\mathfrak{J})$ increases with $1/T$, the smaller factor $g$ is required when the temperature is lower in order to justify the Born-Markov approximation, which reduces the efficiency of energy dissipation. With these parameters $(T,g,t,\sigma,M)$, we run a QPE trial batch. After the QPE batch, we update the parameters by taking $T \leftarrow T/2$, $g \leftarrow \lambda_1 g/2$, $\sigma \leftarrow 2\lambda_5\sigma$ and $M \leftarrow 4M/\lambda_1^2$. Here, $\lambda_1,\lambda_4,\lambda_5$ are constants given in Appendix \ref{app: iteration method}, in which we also explain why we update parameters in this way. With the updated parameters, we run another QPE batch; and we repeat the parameter updating-QPE batch loop. We use $L$ to denote the total number of QPE batches, i.e. the parameters are updated for $L-1$ times. 

We stop the parameter-updating loop when confident about $\hat{E}_{\text{min}}$. In practical classical algorithms such as Quantum Monte Carlo method \cite{carlson2015quantum}, Lanczos method \cite{prelovvsek2013ground} and density-matrix renormalization group method \cite{schollwock2005density}, a broadly used approach is stopping the computing when the result becomes stable. We can take the same approach to stop the parameter-updating loop by tracking the value of $\hat{E}_{\text{min}}$. 
	
	\section{Robustness of the algorithm}
	\label{sec: robustness}
	
	The efficiency of the algorithm above depends on whether the super bath induces efficient energy dissipation. We call a bath that induces efficient energy dissipation a good bath. Intuitively, any super-bath that contains this good bath should possess the same capability to dissipate energy,suggesting the existence of a partial order relation among baths: for baths $A$ and $B$, we say  $A \preceq B$ if, for any Hamiltonian that $A$ can induce to the ground state in polynomial time, $B$ can do so as well.

	In this section, we demonstrate that this partial order relation holds under the condition that the rotating wave approximation is valid, thereby allowing the system's evolution to be described by the Lindblad equation. Therefore, if the super-bath is good, we can efficiently compute the ground state energy. On the contrary, we also show that (up to certain conditions), when the super-bath is not good, then a good bath does not exist. In other words, if there exists a physical bath able to efficiently dissipate the system energy in nature, then the super-bath construction provides a good chance to achieve the same result.

	We quantify the efficiency of energy dissipation with its power. Here, we assume the Born-Markov approximation holds, which allows us to intutively justify the algorithm's properties. In the following subsection, we are going to lift such a strong assumption by calculating the error of Redfield equation.  Under the Born-Markov approximation, the {dissipated} power reads  
	\begin{eqnarray}
		P_\mathcal{K}(T,J;\rho) = -\Tr[H\mathcal{K}(T,J;0)\rho],
		\label{eq:power}
	\end{eqnarray}
	where $\rho$ is the state of the system and subscript $\mathcal{K}$ represents the dissipation term in the Redfield equation. {We are now going to characterise this dissipation by taking the minimum power in a set of states and a range of the temperature, as follows. }
	
	\begin{definition}
		{\bf Minimum power.} Given a spectral density $J$ and a temperature supremum $1/\beta$, the minimum power is 
		\begin{eqnarray}
			P_{\mathcal{K},\text{min}}(1/\beta,J) &=& \min\{ P_\mathcal{K}(T,J;\rho) \vert \notag \\
			&& T\in (0,1/\beta], \notag \\
			&& \rho\in \mathbb{S},\bra{G}\rho\ket{G}\leq 1/2 \},
		\end{eqnarray}
		where $\mathbb{S}$ denotes the set of normalized reduced density matrices of the system. 
		\label{def:MP}
	\end{definition}

	In the above definition, we take the minimum in states with $\bra{G} \rho \ket{G} \leq 1/2$. Without this condition on states, the minimum is always (approximately) zero because the power vanishes when the state approaches the ground state. Notice that $\bra{G} \rho \ket{G} \leq 1/2$ is sufficient for observing the ground state in QPE by taking a sufficiently large batch size $K$. 
	{ As mentioned, the main assumption which the presented algorithm relies on, is the existence of a good bath. Explicilty, we define a good bath to have a non-zero minimu power, i.e., 
	\begin{equation}
	P_{\mathcal{K},\text{min}}(1/\beta,J_{\text{good}}) > 0\;.
	\end{equation}
	We are now going to address the following question: Does the existence of a good bath (with spectral density $J_\text{good}$) allow to provide a lower bound on the dissipative capabilities of the super-bath (with spectral density $J_\text{S}$) defining the channel for the open dynamics? To address this, we define a complementary bath with spectral density $J_{\text{C}}=J_{\text{S}} - \eta J_{\text{good}}$, in terms of a positive number $\eta\geq 0$. In other words
	\begin{equation}
	\label{eq:dec}
	J_\text{S}=J_\text{C} + \eta J_{\text{good}}\;.
	\end{equation}

	This complementary bath is physical when $J_{\text{C}}(\omega)$ is semi-definite positive, which ultimately requires $\eta \leq \mathfrak{J}(\omega)/\norm{J_\text{good}(\omega)}_\infty$, where $\norm{\bullet}_\infty$ is the spectral norm. For example, a conservative choice is to consider $\eta = \min_{\omega\geq 0} \mathfrak{J}(\omega)/\norm{J_\text{good}(\omega)}_\infty$.  } 
	
	 In order to make progress and to provide an intuitive explanation of the main algorithm, here we further assume the rotating-wave approximation so that all the effects of the bath we considered can be described in terms of Lindblad equations. This assumption will be relaxed in the later discussions.
	Under the rotating-wave approximation, the dissipation power for the super-bath can be written as 
	\begin{eqnarray}
		P_\mathcal{L}(T,J_{\text{S}};\rho) = -\Tr[H\mathcal{L}(T,J_{\text{S}})\rho_s]\;.
	\end{eqnarray}
	According to the decomposition above, we further have
	\begin{equation}
	P_\mathcal{L}(T,J_{\text{S}};\rho) = \eta P_\mathcal{L}(T,J_\text{good};\rho)+P_\mathcal{L}(T,J_{\text{C}};\rho)\;.
	\end{equation}
	Under the definition of $a$ provided above, the complementary bath is physical so that, $P_\mathcal{L}(T,J_{\text{C}};\rho)\geq 0$ at low temperatures such that thermal transitions increasing the energy are sufficiently suppressed.  This implies that 
	\begin{equation}
		\label{eq:PaP}
	P_\mathcal{L}(T,J_{\text{S}};\rho) \geq \eta P_\mathcal{L}(T,J_\text{good};\rho)\;,
	\end{equation}
	which is also positive since $\eta \geq 0$ and by the definition of a good bath. This shows how the existence of a good bath allows to conclude a positive dissipation power for the super-bath in terms of the efficiency parameter $\eta$.

	\begin{figure}[hbtp]
		\centering
		\captionsetup[subfigure]{justification=raggedright, singlelinecheck=false}
		\begin{subfigure}[t]{0.45\textwidth}
			\centering
			\caption{} 
			\includegraphics[width=\columnwidth]{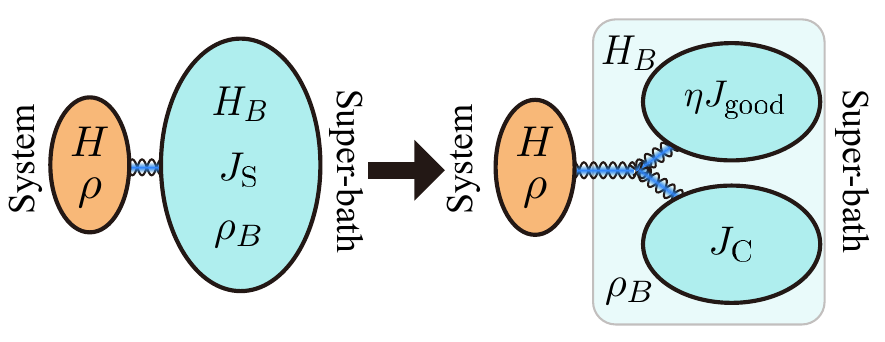}
			\label{fig:a}
		\end{subfigure}
		\hfill
		\begin{subfigure}[t]{0.45\textwidth}
			\centering
			\caption{} 
			\includegraphics[width=\columnwidth]{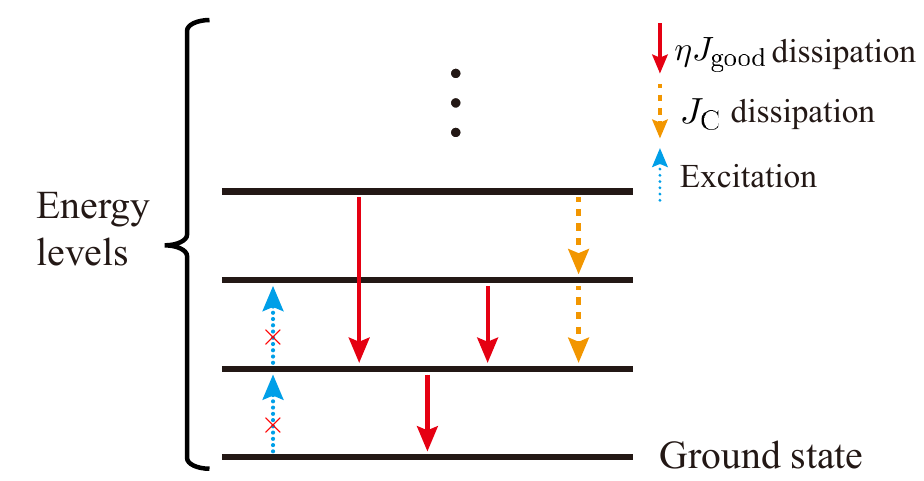}
			\label{fig:b}
		\end{subfigure}
		\captionsetup{justification=Justified, singlelinecheck=false, font=small}
		\caption{(a) The coupling between the system and the super-bath can be decomposed into two parts, 
			$\eta J_{\text{good}}$ and $J_{\text{C}}$. (b) System transitions triggered by the coupling to the super-bath. The red arrows represent energy dissipation caused by the good bath, driving the system to the ground state. The yellow dashed arrows represent energy dissipation caused by the complementary bath, which lowers the energy but may not lead to the ground state. The blue dotted arrows represent excitation caused by the bath. When the bath is at zero temperature, the excitation is prevented according to the Lindblad equation.}
		\label{fig:transition}
	\end{figure}
	
	 As mentioned, this result relies on the rotating-wave approximation which, in general, might not hold. We addressed this by the Gaussian stabilization channel (see Definition \ref{def:GS}), which prepares an approximate stationary state, and allows to lift the rotating-wave assumption.

	 Our key observation is that if the state $\rho_s$ is a stationary state, both the Redfield equation and the Lindblad equation predict the same energy dissipation power, i.e. $P_\mathcal{K}(T,J_\text{S};\rho_s) = P_\mathcal{L}(T,J_\text{S};\rho_s)$.  To see this, we can write $H = \sum_m E_m \Pi_m$ as the spectral decomposition of the system Hamiltonian,  in terms of its eigenvalues $E_m$  and corresponding eigenspace projectors $\Pi_m$.	With this notation, a stationary state can be written as $\rho_s = \sum_m \rho_m$, where $\rho_m$  belongs to the subspace  defined by $\Pi_m$, i.e. $\rho_m = \Pi_m\rho_m\Pi_m$.  Stationary states evolve trivially under the system Hamiltonian. Taking one term from each of $\mathcal{K}$ and $\mathcal{L}$ as an example, the energy dissipation rates are
	\begin{eqnarray}
		P_\mathcal{K}(T,J_\text{S};\rho_s) = -\Gamma_{\alpha,\beta}(\omega)D_{\alpha,\beta}(\omega,\omega') + \cdots
	\end{eqnarray}
	and 
	\begin{eqnarray}
		P_\mathcal{L}(T,J_\text{S};\rho_s) = -\delta_{\omega,\omega'}\Gamma_{\alpha,\beta}(\omega)D_{\alpha,\beta}(\omega,\omega') + \cdots,
	\end{eqnarray}
	where $D_{\alpha,\beta}(\omega,\omega') = \Tr[HA_\beta(\omega)\rho_s A_\alpha^\dag(\omega')]$, and 
	\begin{eqnarray}
		A_\alpha(\omega) = \sum_{m,n\vert E_n-E_m=\omega} \Pi_m A_\alpha \Pi_n.
	\end{eqnarray}
	 Now, since  $\rho_s$ is a stationary state, we have that $D_{\alpha,\beta}(\omega,\omega') = \delta_{\omega,\omega'}D_{\alpha,\beta}(\omega,\omega')$, which holds for all the terms present in the decompositions for both 
	$\mathcal{K}$ and $\mathcal{L}$.  This implies that, when applied to stationary states, the dissipation powers are the same, i.e., 
	\begin{equation}
	P_\mathcal{K}(T,J_\text{S};\rho_s) =P_\mathcal{L}(T,J_\text{S};\rho_s) \;.
	\end{equation}
	
 In conclusion, for the energy dissipation power of a stationary state, the Redfield equation and the Lindblad equation are equivalent, regardless of the rotating-wave approximation.  This allows to conclude that the bound in Eq.~(\ref{eq:PaP}) is valid even without the requirement of a rotating-wave approximation.

	As a summary, we have intuitively demonstrated the cooling efficiency of the super-bath under the rotating wave and Lindblad approximations. By dividing the superbath into a good bath with efficient energy dissipation and a complementary generic bath, we showed that the super-bath and the good bath have the same energy dissipation efficiency, up to an efficiency  factor. In what follows, we show that even when the rotating wave approximation does not hold, this conclusion remains valid within our optimized algorithmic framework, as long as the parameters lie within an appropriate range.

	\subsection*{Main result}
	
In this subsection, we provide a rigorous derivation of the results intuitively argued in the previous section, without the need of extra assumption such as the rotating wave approximation. To do this, we are going to introduce the concept of a `sub-bath', and use it to prove Lemma \ref{lem:super} and, based on this, demonstrate the efficiency of our algorithm.
	
	\begin{definition}
		{\bf Sub-bath spectral density.} A spectral density  $J_{\text{sub}}$ is said to be a sub-bath spectral density with a constraint factor of $b$ if and only if 
		\begin{itemize}
			\item[1.] For all $\omega$,  $J_{\text{sub}}(\omega)$ is positive semi-definite; 
			\item[2.] For all $\omega$, 
			\begin{eqnarray}
				\norm{ J_{\text{sub}}(\omega)}_\infty\leq b\mathfrak{J}(\omega);
			\end{eqnarray}
			\item[2.] For all $T>0$, 
			\begin{eqnarray}
				\label{eq: 4}
				\tau_R(T,J_{\text{sub}}) &\geq& b^{-1}\tau_R(T,\mathfrak{J}), \\
				\label{eq: 5}
				\tau_B(T,J_{\text{sub}}) &\leq& b\tau_B(T,\mathfrak{J}).
			\end{eqnarray}
		\end{itemize}
		\label{def:subBSD}
	\end{definition}

	In this definition, the first condition states that the sub-bath spectral density is a valid spectral density. The second condition corresponds to assume the existence of a decomposition of the  super bath $J_{\text{S}}$ into a sub-bath $J_{\text{sub}}$ and a complementary bath $J_{\text{C}}$ in terms of an efficiency factor $b$. In turn, this property allow us to clarify the logic behind the decomposition in Eq.~(\ref{eq:dec}). If we assume the existence of a good bath which is also a sub-bath, then the decomposition in Eq.~(\ref{eq:dec}) holds for all $\eta\leq 1/b$. Under the third condition, the Born-Markov approximation can be applied to the bath $J_\text{sub}$ with an error comparable to the error for the super bath (depending on $b$). If the factor $b$ is larger, the component of  $J_\text{sub}$  in the super bath is smaller (the energy dissipation power of the super bath is potentially smaller), and the error in the Refiled equation for the bath  $J_\text{sub}$  is potentially larger. 
	
	According to our previous analysis, we expect the super bath to cause the system to dissipate energy at a power of $(g^2/b)P_{\mathcal{K},\text{min}}(1/\beta,J_{\text{S}})$ under certain approximations.  The lifting of these approximations may potentially reduce the total power of energy dissipation and we are going to take this into account. For example, there is an error between the exact time evolution and the Redfield equation (note that the energy dissipation power defined in Eq. (\ref{eq:power}) is based on the Redfield equation). Second, Eq. (\ref{eq:power}) defines the instance energy dissipation power, and there is an error when we use it for a finite time interval. Third, the state $\mathcal{G}_\sigma\rho$ is not exactly a steady state. Fourth, even if the initial state $\mathcal{G}_\sigma\rho$ is a steady state, the state may deviate from steady states during the time evolution. Finally, the complementary bath may excite the system due to the finite temperature. We  take into account all these potential errors in the following lemma and illustrate that we can bound the total error by scaling the parameters of the algorithm.  the total error by taking proper parameters. The proof is given in Appendix \ref{app: proof of lemma1}, in which we provide more details on the scaling of the parameters. 
	
	\begin{lemma}
	\label{lem:super}
	Suppose there exists a good sub-bath spectral density $J_{\textup{g-sub}}$ with a constraint factor of $b$ and a minimum power $P_{\mathcal{K},\textup{min}}(1/\beta,J_{\textup{g-sub}})$. For all states $\rho\in \mathbb{S}$ satisfying $\bra{G}\rho\ket{G}\leq 1/2$, the energy reduced in one round of the cooling cycle has a lower bound 
	\begin{eqnarray}
		\Tr(H\rho) - \Tr[H\mathcal{M}(H,T,g^2J_{\textup{S}};t)\mathcal{G}_\sigma\rho] \geq& Pt - \epsilon,
		\label{eq:dissipation_bound}
	\end{eqnarray}
	where $P = (g^2/b)P_{\mathcal{K},\textup{min}}(1/\beta,J_{\textup{g-sub}})$ . The error term has an upper bound $\epsilon\leq Pt/2$ if we take proper parameters 
	\begin{eqnarray}
		1/T,1/g,\sigma = O(\mathrm{Poly}(N,b,r,h,\beta)),
	\end{eqnarray}
	where $N = \abs{\mathbb{A}}$ is the size of the operator set, $r = 1/P_{\mathcal{K},\textup{min}}(1/\beta,J_{\textup{g-sub}})$   and $h = \norm{H}_\infty$. 
	\end{lemma}
	
	According to the above lemma, the super-bath algorithm is efficient provided an efficient sub-bath exists, which demonstrates the partial order relation among baths. The only problem is to find the proper parameters. However, since these parameters are also shown to scale polynomially in the system size, they are expected to be found efficiently by the update procedure of the algorithm. We note that some constraints exist among the parameters (see Appendix \ref{app: proof of lemma1}). In the algorithm, the parameter updating rule is designed to satisfy the constraints. 
	
	With the above observation about finding proper parameters in the parameter updating loop, we can analyse the time complexity of the algorithm. The overall complexity is the product of a few factors, including the number of QPE bathes $L$ (i.e. parameters are updated for $L-1$ times), the size of each bath $K$, the number of cooling cycles in each QPE trial $M$ and the time complexity of each cooling cycle (increasing with $t$ and $\sigma$). Notice that when we update parameters, the time complexity of the QPE batch increases. Therefore, we focus on the last round of QPE batch for the factors $M$, $t$ and $\sigma$. We summarise these factors in the following theorem. See Appendix \ref{app: proof of lemma1} for the proof. 
	
	\begin{theorem}
	Let $1/\beta$ be the temperature such that 
	\begin{eqnarray}
		\frac{\bra{G} e^{-\beta H} \ket{G}}{\Tr_B(e^{-\beta H})} = \frac{3}{4}.
		\label{eq:beta}
	\end{eqnarray}
	Suppose there exists a sub-bath spectral density $J$ with a constraint factor of $b$ and a minimum power $P_{\mathcal{K},\textup{min}}(1/\beta,J)$. Let $\kappa$ be a positive number. In the last round of the QPE batch, the ground state can be observed with a probability larger than $1-\kappa$ with 
	\begin{eqnarray}
		L = O(\mathrm{Polylog}(N,b,r,h,\beta)).
	\end{eqnarray}
	The size of each QPE batch is 
	\begin{eqnarray}
		K = O(\log(\kappa)).
	\end{eqnarray}
	In the last round of the QPE batch, 
	\begin{eqnarray}
		M,t,\sigma = O(\mathrm{Poly}(N,b,r,h,\beta)).
	\end{eqnarray}
	\label{the:complexity}
	\end{theorem}
	
	The theorem is proved by demonstrating that in the last round of QPE batch, the probability of the ground state $\bra{G} \rho_{\text{f}} \ket{G}$ is not smaller than $1/2$. In the theorem, the choice of ground state probability $3/4$ is arbitrary. The only requirement is that it is larger than $1/2$. The energy dissipation usually slows down when the state approaches equilibrium. Therefore, a difference between the two probabilities is necessary in order to maintain the finite dissipation power. \\

	{\bf How to choose the super-bath spectral density function.} We now explain how to select an appropriate super-bath spectral density function for a physical system to be solved. Physical spectral density function is typically constructed as the product of a function describing its low-frequency characteristics (e.g., Ohmic, super-Ohmic, sub-Ohmic spectral density) and a high-frequency cutoff function (e.g., Lorentz-Drude cutoff, exponential cutoff, Gaussian cutoff, etc.). While the low-frequency characteristics of the spectral density function primarily determines the dynamics of an open quantum system \cite{breuer2002theory}, the specific form of the super-bath spectral density function must be tailored to the physical model under consideration. For example, in nuclear physics, the selection of an appropriate super-bath spectral density function can be guided by gamma decay transition probabilities. For gamma emission of frequency of $\omega$, the transition probability per second is
	\begin{eqnarray}
		T(l;R) = \frac{8\pi(l+1)}{l[(2l+1)!!]^2}\frac{\omega^{2l+1}}{c^{2l+1}\hbar}B(Rl),
	\end{eqnarray}
	where $c$ represents the speed of light; $l$ represents the multipole order, defining the type of transition and its angular momentum characteristics (e.g., $l = 1$ for dipole, $l = 2$ for quadrupole); $R$ stands for $E$ or $M$, indicating whether the transition is electric ($E$) or magnetic ($M$); and $B(Rl)$ represents the reduced transition probability, depending on nuclear structure and initial and final states involved in the gamma decay \cite{krane1991introductory, greiner1996nuclear}. From this formula, we see that the spectral density can be designed as a linear combination of such terms, which encompasses all possible real physical processes.
	
	Our algorithm has the potential to bypass the constraints of real physical processes. In nuclear physics, certain metastable states exhibit long lifetimes because selection rules based on angular momentum and parity conservation forbid low-order multipole transitions, leaving only small-probability high-order multipole transitions. Since our algorithm simulates the dynamics of an open quantum system on a quantum computer, it is not restricted by such selection rules. This flexibility may allow us to design a super-bath spectral density function with a low-frequency component proportional to $\omega^3$, potentially enabling more efficient energy dissipation and faster convergence to the ground state. 
	
	 Furthermore, since we have established the partial order relation among baths, every super-bath we consider represents a group of baths.
	\section{Implementation issues and scalability}
	\label{sec:implementation and scalability}
	
	There are three quantum subroutines utilized in the super-bath algorithm: the Gaussian stabilisation $\mathcal{G}_\sigma$, the open system simulation $\mathcal{M}(H,T,g^2J_{\text{S}};t)$ and QPE. The Gaussian stabilisation is a Hamiltonian simulation of the system with a randomised evolution time. The open system simulation is a Hamiltonian simulation of the entire system, including the system and the bath. Therefore, the super-bath algorithm is essentially realised through two standard quantum subroutines, Hamiltonian simulation \cite{childs2012hamiltonian,berry2015hamiltonian,low2019hamiltonian,mc2023classically} and QPE  \cite{kitaev1995quantum, dorner2009optimal, nielsen2010quantum, Wan2021A, clinton2024quantum}. Notably, QPE itself also relies on efficient Hamiltonian simulation. Recently, reference \cite{clinton2024quantum} presents a control-free quantum phase estimation method that significantly reduces quantum circuit depth. 
	
	There are various algorithms of the Hamiltonian simulation. Although simulating unitary dynamics generated by arbitrary Hamiltonians remains challenging even on quantum computers, for a substantial class of physical systems, the time complexity of Hamiltonian simulation on a quantum computer scales polynomially with respect to the system size and permissible error \cite{nielsen2010quantum}, such as for Hamiltonians of molecules, nuclei, and many condensed matter models \cite{loaiza2023reducing,holland2020optimal,clinton2024quantum}.  In the open system simulation, we also need to model the bosonic bath. There have been recent advances in both theoretical and experimental aspects of this field. Ref. \cite{PhysRevLett.134.070604} proposes a gate-based quantum simulation framework for Gaussian bosonic circuits that operates on exponentially many modes without requiring direct simulation of infinite-dimensional Hilbert spaces, effectively reducing complex bosonic evolutions to manageable qubit-based quantum computations. Reference \cite{PhysRevA.109.062402} and \cite{sun2024quantumsimulationspinbosonmodels}  demonstrate the simulation of the spin-boson model in trapped-ion systems, enabling precise control over bath temperature and spectral densities for studying dissipative quantum dynamics. 
	
	\begin{corollary}
	Let $n$, $\kappa$ and $\epsilon$ be the size of the system, the permissible failure probability and the permissible error. The time (qubit) complexity for observing the ground state in the super-bath algorithm is $O(\mathrm{Poly}(n,1/\kappa,1/\epsilon))$ under the following conditions: 
	\begin{itemize}
		\item[1.] The spectrum range of the Hamiltonian is $h = O(\mathrm{Poly}(n))$; 
		\item[2.] The temperature in Eq. (\ref{eq:beta}) is $1/\beta = O(\mathrm{Poly}(n))$; 
		\item[3.] There exists a good sub-bath $J_{\textup{g-sub}}$ with a constraint factor $b$ satisfying $1/P_{\mathcal{K},\textup{min}}(1/\beta,J_{\textup{g-sub}})$ and $b = O(\mathrm{Poly}(n))$; 
		\item[4.] The operator $e^{-iHt}$ can be implemented with an additive error $\epsilon_1$ at the time (qubit) cost $\mathrm{Poly}(n,t,1/\epsilon_1)$; 
		\item[5.] The superoperator $\mathcal{M}(H,T,g^2J_{\text{S}};t)$ can be implemented with an additive error $\epsilon_2$ at the time (qubit) cost $\mathrm{Poly}(n,1/T,1/g,t,1/\epsilon_2)$. 
	\end{itemize}
	\label{the:condition}
\end{corollary}
	
	The time for attaining the ground state depends on the energy dissipation power and the initial energy of the system. Condition 1 introduces a bound on the initial energy of the system. Under condition 2, we can reach a sufficiently low temperature in the parameter updating loop. Condition 2 also implies an energy gap $O(\mathrm{Poly}(n))$ between the ground state and the first excited state. This gap is required in QPE. With the gap, QPE cannot distinguish the ground state from the first excited state. The gap guarantees that QPE can successfully project the final state onto the ground state.

	\subsection*{Comparison with previous work}
	
	Numerous studies have applied open system theory in designing quantum algorithms for solving ground states, typically by preparing a thermal state at low temperatures. Reference \cite{temme2011quantum} introduces a quantum generalization of the classical Metropolis sampling algorithm, which uses random unitary transformations on the system state. The acceptance or rejection of each transformation follows the Metropolis rule, and the concept of ``quantum detailed balance'' ensures that the thermal state is the unique fixed point. Other improved versions of this method have also been proposed 
\cite{yung2012quantum, moussa2019low, Lemieux2020, Layden2023, Jiang2024}. 
	
	A more physically motivated approach involves simulating the dynamics of a system coupled to a low-temperature heat bath. The ground state (or more generally, a low-temperature thermal state) should be readily achievable by simply coupling the system to a low-temperature bath. Building on this intuition, several studies have proposed efficient methods for preparing a system in a low-temperature thermal state by simulating its evolution under interactions with an appropriately designed environment \cite{Terhal2000, Shabani2016, Kaplan2017, Metcalf2020, Polla2021, Rost2021, movassagh2023preparing, Chen2021fast, Metcalf2022, Cattaneo2023, Mi2024}.
	
	Since the evolution of open quantum systems can be approximately described by the Lindblad equation under certain conditions \cite{davies1974markovian, lindblad1976generators}, simulating the Lindblad dynamics provides another approach to prepare the ground state or low-energy thermal state \cite{Kliesch2011, Childs2017, Cleve2016, Li2022, pocrnic2023quantum, Ding2023, Ding2024a, Chen2024,yu2024exponentially}. Note that the simulated Lindblad dynamics may not reproduce any real system-bath dynamics found in nature \cite{Ding2023}. Quantum Gibbs samplers based on Lindblad dynamics have gained significant attention, with notable advancements in the field \cite{wocjan2023szegedy, rall2023thermal, chen2023quantum, chen2023efficient, Ding2024, Gilyen2024}. In these algorithms, mixing time often appears as a parameter in the time complexity. Recently, reference \cite{rouzé2024optimalquantumalgorithmgibbs} provides a lower bound on the temperature that guarantees a rapid mixing time for the Lindbladian introduced in \cite{chen2023quantum}. However, the mixing time at low temperatures remains uncertain, which could result in exponential growth in time complexity when solving for the ground state.
	
	In contrast to previous methods, our approach simulates the dynamics of the system coupled to the super-bath, and additionally incorporates a dephasing process. The dephasing process relaxes the requirement of precisely designing an environment that drives the system toward the ground state, allowing us to employ a super-bath encompassing various dissipation processes. Consequently, we establish a sufficient condition for solving the ground state problem of a physical system on a quantum computer within polynomial time: the existence of a bath bounded by a proper constraint function that allows the system to dissipate into its ground state, even without detailed knowledge of the bath’s specific characteristics.

	\section{Conclusions}
	\label{sec:conclusions}
	
	In this work, we propose the Super-Bath Quantum Eigensolver: a novel algorithm for ground-state preparation of quantum systems.  Under certain conditions, our algorithmic framework ensure a partial order relation among environments. Therefore, the efficiency of this algorithm relies on a single assumption: the existence of a physical sub-bath able to induce a reduced system dynamics converging to the ground state in a time which is polynomial in the system size. A detailed characterization of the bath is not needed. Despite its generality, this assumption allows to define a super-bath which can be practically simulated. More precisely, the proposed Quantum Eigensolver consists of a cooling cycle made by a Gaussian stabilization channel followed by an open dynamics induced by the interaction with the super-bath. We explicitly determine the scaling of the simulation performance in terms of properties of the super-bath spectral density function to show that the original existence assumption results in an efficient cooling algorithm. These findings can find application to a variety of different physical systems. In particular, experimental data for the lifetime of nuclear isomers due to gamma decay suggests that the algorithm could be used to find nuclear ground states. Future work may focus on identifying suitable super-baths for other physical systems, such as atomic ensembles.

	\section*{Acknowledgements}	
	This work is supported by the National Natural Science Foundation of China (Grant Nos. 12225507, 12088101) and NSAF (Grant No. U1930403). We would like to thank Hanxu Zhang for the helpful discussions. 
	
	\appendix
	
	\begin{widetext}
		
		\section{Open quantum systems}
		\label{app:open_system}
		
		In this section, we provide explicit definitions of notations given in Sec. \ref{sec:notations}. We also give the error bound between the density matrix evolved according to the Redfield equation and the exact density matrix of the system in the case of a Gaussian bath. 
		
		For an open system, the system interacts with a bath. The Hamiltonian of the combined system consisting of the system and the bath reads 
		\begin{eqnarray}
			H_{SB} = H\otimes \openone_B + \openone\otimes H_B + H_I, 
		\end{eqnarray}
		where $H$ and $H_B$ ($\openone$ and $\openone_B$) are the Hamiltonians (identity operators) of the system and bath, respectively, and $H_I$ is the interaction between them. Without loss of generality, the interaction can be expressed in the form $H_I = \sum_{\alpha = 1}^N A_\alpha\otimes B_\alpha$, where $A_\alpha$ and $B_\alpha$ are Hermitian system and bath operators, respectively. We choose $A_\alpha$ to be dimensionless with $\norm{A_\alpha}_\infty = 1$ for all $\alpha$, where $\norm{\bullet}_\infty$ denotes the spectral norm.

		Suppose the system and bath are uncorrelated at the time $t = 0$. The state of the system at the time $t$ is 
		\begin{eqnarray}
			\rho(t) = \Tr_{B}\left[e^{-iH_{SB}t}\left(\rho(0)\otimes\rho_B\right)e^{iH_{SB}t}\right], 
		\end{eqnarray}
		where $\Tr_B$ is the partial trace over the bath, $\rho(0)$ is the initial state of the system, and $\rho_B$ is the initial state of the bath. We focus on the case that $\rho_B$ is a thermal state. In an expansion of the time evolution operator $e^{-iH_{SB}t}$ in the interaction picture, one can find that the time evolution of the open system is determined by bath correlation functions \cite{breuer2002theory}. The two-time bath correlation function is 
		\begin{eqnarray}
			C_{\alpha,\beta}(s) = \Tr_B\left[B_\alpha(t)B_\beta(t-s)\rho_B\right],
			\label{eq:correlation}
		\end{eqnarray}
		where $B_\alpha(t) = e^{iH_Bt}B_\alpha e^{-iH_Bt}$ are bath operators in the interaction picture. One can verify that the correlation is independent of $t$ because the bath is in a thermal state. 
		
		\subsection{Gaussian boson bath}
		
		In the case of a Gaussian boson bath, the bath Hamiltonian is in the form 
		\begin{eqnarray}
			H_B = \sum_l\int_0^\infty d\omega \omega b_l^{\dag}(\omega)b_l(\omega),
		\end{eqnarray}
		where $[b_l(\omega),b_{l'}^\dag(\omega')] = \delta_{l,l'}\delta(\omega-\omega')$. Here, $\delta_{l,l'}$ is the Kronecker delta, and $\delta(\omega-\omega')$ is the Dirac delta function. The bath operators are in the form 
		\begin{eqnarray}
			B_\alpha = \sum_l\int_0^\infty d\omega \big(u_{\alpha,l}(\omega)b_l(\omega) + u_{\alpha,l}^*(\omega)b_l^{\dag}(\omega)\big),
		\end{eqnarray}
		where $u_l(\omega)\in \mathbb{C}$ are coupling coefficients. When the bath is in the thermal state at the temperature $T$, i.e. $\rho_B = e^{-H_B/T}/\Tr(e^{-H_B/T})$, the bath correlation functions are 		
		\begin{eqnarray}
			C_{\alpha,\beta}(T,J;s) \equiv  \int_0^\infty  J_{\alpha,\beta}(\omega) \frac{e^{i\omega s}}{e^{\beta\omega}-1} + J_{\alpha,\beta}^*(\omega) e^{-i\omega s} \left(\frac{1}{e^{\beta\omega}-1}+1\right) d\omega,
			\label{eq:correlationGB}
		\end{eqnarray}		
		where 
		\begin{eqnarray}
			J_{\alpha,\beta}(\omega) \equiv \sum_l u^*_{\alpha,l}(\omega)u_{\beta,l}(\omega)
		\end{eqnarray}
		is the spectral density function. For all $\omega$, the matrix $J(\omega)$ is positive semi-definite. We remark that Eq. (\ref{eq:correlationGB}) is derived by applying Eq. (\ref{eq:correlation}) to the case that the bath is Gaussian and $\rho_B$ is a thermal state.

		For Gaussian boson baths, we can derive all multi-time bath correlation functions from the two-time correlation function, i.e. all Gaussian boson baths with the same two-time correlation function lead to the same time evolution of the open system  \cite{fogedby2022field}. Therefore, given a tuple of system operators $(A_1,A_2,\ldots,A_N)$, the time evolution of the open quantum system only depends on the bath temperature $T$ and spectral density $J$ in addition to the system Hamiltonian. The time-dependent trace-preserving completely positive map describing the time evolution of the open system in the Schrödinger picture reads 
		\begin{eqnarray}
			\mathcal{M}(H,T,J;t) \equiv \Tr_{B}\left[e^{-iH_{SB}t}\left(\bullet\otimes\rho_B\right)e^{iH_{SB}t}\right],~~~
		\end{eqnarray}
		where $H_B$ and $B_\alpha$ are taken such that the spectral density is $J$.

		\subsection{Error in the Redfield equation}
		\label{app: error}
		
		The generator in the Redfield equation reads 
		\begin{eqnarray}
			\mathcal{K}(T,J;t)\bullet &\equiv& \sum_{\alpha,\beta}\int_0^\infty ds C_{\alpha\beta}(T,J;s) \Big(A_\beta(t-s)\bullet A_\alpha(t) - A_\alpha(t)A_\beta(t-s)\bullet\Big) + H.c.,
		\end{eqnarray}
		where $A_\alpha(t) = e^{iHt}A_\alpha e^{-iHt}$ are system operators in the interaction picture. This equation is approximate, and the approximation is characterised by two quantities 
		\begin{eqnarray}
			\label{eq: tauR}
			\tau_R(T,J) \equiv \Big(4\int_0^\infty ds \norm{C(T,J;s)}_{L_1}\Big)^{-1}
		\end{eqnarray}
		and 
		\begin{eqnarray}
			\label{eq: tauB}
			\tau_B(T,J) \equiv 4\tau_R(T,J)\int_0^\infty ds s\norm{C(T,J;s)}_{L_1},
		\end{eqnarray}
		corresponding to the timescales of system relaxation and bath correlation, respectively. Here, $\norm{\bullet}_{L_1}$ denotes the entry-wise matrix $L_1$ norm, i.e. $\norm{C}_{L_1} = \sum_{\alpha\beta}\abs{C_{\alpha\beta}}$. 
		
		\begin{lemma}
			\label{lem: Redfield error}
			There is a rigorous upper bound on the error in the Redfield equation $\frac{d}{dt}\rho_R(t) = -i[H,\rho_R(t)] + \mathcal{K}(T,J;0) \rho_R(t)$. For all $\rho(0)\in \mathbb{S}$ and $0 \leq t \leq \tau_R\ln\left(1 + \frac{\tau_R}{2\tau_B}\right)$, the inequality 
			\begin{eqnarray}
				\norm{\rho_R(t) - \rho(t)}_1 \leq 2\left(e^{t/\tau_R}-1\right)\frac{\tau_B}{\tau_R} + \varepsilon_l(t)
			\end{eqnarray}
			holds, where
			\begin{eqnarray}
				\varepsilon_l(t) = \left\{\begin{array}{ll}
					2 \left(e^{t/\tau_R}-1\right), & t \leq \tau_B, \\
					2 e^{t/\tau_R}\left[\left(1-e^{-\tau_B/\tau_R}\right)+\frac{\tau_B}{\tau_R}\ln{\frac{t}{\tau_B}} \right], & t > \tau_B.
				\end{array}\right.
			\end{eqnarray}
			Here, $\rho(t)=\mathcal{M}(H,T,J;t)\rho(0)$ is the exact state of the open system, $\rho_R(t)$ is the solution to the Redfield equation taking the initial state $\rho_R(0) = \rho(0)$, and $\norm{\bullet}_1$ denotes the trace norm. 
		\end{lemma}

		\subsubsection{Derivation of Redfield equation and error terms}
		\label{subsection: derivation of Redfield equation}
		We briefly outline the derivation of the Redfield equation to highlight the error terms introduced by approximations such as the Born-Markov approximation.

		The Redfield equation is most easily derived in the interaction picture. We transform to the interaction picture by applying a unitary transformation generated by the Hamiltonian $H + H_B$. After this transformation, the Hamiltonian of the combined system in the interaction picture is given by
		\begin{eqnarray}
			{H}_I(t) = \sum_{\alpha=1}^N {A}_\alpha(t) \otimes {B}_\alpha(t),
		\end{eqnarray}
		where ${A}_\alpha(t) = e^{-iHt}A_\alpha e^{iHt}$ and ${B}_\alpha = e^{-iH_Bt}B_\alpha e^{iH_Bt}$ are system and bath operators in the interaction picture respectively. The von Neumann equation of the combined system can be written as 
		\begin{eqnarray}
			\label{eq: von Neumann equation}
			\frac{d}{dt}\rho_{SB}(t) = -i[H_I(t),\rho_{SB}(t)].
		\end{eqnarray}
		Integrating once, substituting the result back into the original equation, and taking the partial trace over the bath degrees, we obtain a differential equation of the system density matrix,
		\begin{eqnarray}
			\frac{d}{dt} \rho(t) = -\int_0^t ds \Tr_B[H_I(t),[H_I(s),\rho_{SB}(s)]].
		\end{eqnarray}
		Note that the above equation is exact.
		
		It is common in nature that systems are weakly coupled to a large bath: the influence of the system on the bath is small, and the timescale over which the state of the system varies appreciably is large compared to the timescale over which the bath correlation functions decay \cite{breuer2002theory}. The Born-Markov approximation can be used to derive a master equation of such open systems. The Born approximation amounts to using $\rho(t)\otimes \rho_B$ instead of $\rho_{SB}(t)$ inside the integral since the interaction is weak and the bath is large compared to the system. We obtain
		\begin{eqnarray}
			\frac{d}{dt}\rho(t) = \sum_{\alpha,\beta}\int_0^t dsC_{\alpha,\beta}(t-s)\big(A_\beta(s)\rho(s)A_\alpha(t) - A_\alpha(t)A_\beta(s)\rho(s)\big) + H.c. + \xi_B(t),
		\end{eqnarray}
		where the correlation function $C_{\alpha,\beta}(t-s)$ is defined in Eq. (\ref{eq:correlation}), and we use $\xi_B(t)$ to denote the error induced by the Born approximation.
		
		We already define the characteristic system relaxation and bath correlation time in Eq. (\ref{eq: tauR}) and Eq. (\ref{eq: tauB}). Here, we point out that when $\tau_B \ll \tau_R$, the Markov approximation can be implemented by replacing $\rho(s)$ with $\rho(t)$. We obtain
		\begin{eqnarray}
			\frac{d}{dt}\rho(t) = \sum_{\alpha,\beta}\int_0^t dsC_{\alpha,\beta}(t-s)\big(A_\beta(s)\rho(t)A_\alpha(t) - A_\alpha(t)A_\beta(s)\rho(t)\big) + H.c. + \xi_B(t) + \xi_M(t),
		\end{eqnarray}
		where we use $\xi_M(t)$ to denote the error induced by the Markov approximation. We use $\rho_{BM}(t)$ to denote the state of the system derived by the Born-Markov approximation master equation, that is, 
		\begin{eqnarray}
			\label{eq: Born-Markov approximation master equation}
			\frac{d}{dt}\rho_{BM}(t) = \sum_{\alpha,\beta}\int_0^t dsC_{\alpha,\beta}(t-s)\big(A_\beta(s)\rho(t)A_\alpha(t) - A_\alpha(t)A_\beta(s)\rho(t)\big) + H.c..
		\end{eqnarray}
		Note that $\rho_{BM}(t)$ may not be a strict density matrix since the Born-Markov approximation master equation may not correspond to a completely positive and trace-preserving dynamical map.
		
		For Gaussian boson baths, there is a rigorous upper bound on the error in the Born-Markov approximation which will be discussed later.
		
		Finally, we substitute $s$ by $t-s$ int the integral of Eq. (\ref{eq: Born-Markov approximation master equation}) and let the upper limit to the integral go to infinity to obtain a Markovian quantum master equation, i.e., the Redfield equation. Use $\xi_l(t)$ to denote the error caused by changing the upper limit of the integral. We obtain
		\begin{eqnarray}
			\frac{d}{dt}\rho_{BM}(t) = \mathcal{K}(t)\rho_{BM}(t) + \xi_l(t),
		\end{eqnarray}
		where 
		\begin{eqnarray}
			\xi_l(t) &=& -\sum_{\alpha,\beta}\int_t^\infty ds C_{\alpha\beta}(s) \big(A_\beta(t-s)\rho_{BM}(t) A_\alpha(t) - A_\alpha(t)A_\beta(t-s)\rho_{BM}(t)\big) + H.c.
		\end{eqnarray}
		We also analyze this error term later in this section.

		\subsubsection{Born-Markov approximation error}
		
		We first bound the error induced by the Born-Markov approximation. We follow the same approach as in Ref. \cite{nathan2020universal} and obtain similar results. The only difference is that in Ref. \cite{nathan2020universal} the spectral norm is used, whereas we use the trace norm. We demonstrate that using different norms does not affect the final results. To achieve this, two key inequalities are employed. First, since the partial trace $\Tr_B$ is a completely positive and trace-preserving map, the induced trace norm of $\Tr_B$ is less than 1. That is, for any operator $O$ of the Hilbert space of combined system $SB$, we have
		\begin{eqnarray}
			\label{eq: partial trace norm}
			\norm{\Tr_B(O)}_1 \leq \norm{O}_1.
		\end{eqnarray}
		Second, for any operator $O_1,O_2,O_3$ that can be multiplied together, we have
		\begin{eqnarray}
			\label{eq: norm submultiplicative}
			\norm{O_1 O_2 O_3}_1 \leq \norm{O_1}_\infty \norm{O_2}_1 \norm{O_3}_\infty 
		\end{eqnarray}
		
		In order to better visualize the consistency of the derivation process, we adopt the same notation for superoperators as in Appendix A in Ref. \cite{nathan2020universal}, that is, by adding a hat symbol $\hat{\bullet}$ on the superoperators. However, to generalize the spectral norm to the trace norm, we still use the operator form of the density matrix (i.e. $\rho$) rather than the vectorized version (i.e. $|\rho\rangle\rangle$). Then, the von Neumann equation of the combined system in Eq. (\ref{eq: von Neumann equation}) can be rewritten as 
		\begin{eqnarray}
			\label{eq: von Neumann equation in universal}
			\frac{d}{dt}\rho_{SB}(t) = -i\hat{\mathcal{H}}(t)\rho_{SB}(t),
		\end{eqnarray}
		where $\hat{\mathcal{H}}(t) = [H(t),\bullet] = \hat{H}^l(t) - \hat{H}^r(t)$. Here, we define the left- and right-multiplication superoeprator $\hat{A}^l$ and $\hat{A}^r$ as
		\begin{eqnarray}
			\hat{A}^l \rho = A \rho, \qquad  \hat{A}^r \rho =  \rho A.
		\end{eqnarray}
		
		Suppose the initial state is a direct product state,  i.e., $\rho_{SB}(0) = \rho(0) \otimes \rho_B$, the state of the total system at time $t$ is given by
		\begin{eqnarray}
			\rho_{SB}(t) = \hat{\mathcal{U}}(t,0)\big(\rho(0)\otimes \rho_B\big).
		\end{eqnarray}
		Here, $\hat{\mathcal{U}}(t,0)$ denotes the unitary evolution superoperator of the combined system, given by $\hat{\mathcal{U}}(t,0)(\bullet) = \mathcal{T}e^{-i\int_0^t ds \hat{\mathcal{H}}(s)}\bullet$, where $\mathcal{T}$ denotes the time-ordering operation. The reduced density matrix of the system can be rewritten as 
		\begin{eqnarray}
			\label{eq: exact final state in universal}
			\rho(t) = \Tr_B\Big(\hat{\mathcal{U}}(t,0)\big(\rho(0)\otimes\rho_B\big)\Big).
		\end{eqnarray}
		Taking the time derivative, we obtain the time evolution equation of the system-reduced density matrix in the superoeprator form
		\begin{eqnarray}
			\label{eq: A7}
			\frac{d}{dt}\rho(t) = -i \sum_{m,\alpha}\nu_m\hat{A}_\alpha^m(t)\Tr_B\Big(\hat{B}^m_\alpha(t)\hat{\mathcal{U}}(t,0)\big(\rho(0)\otimes\rho_B\big)\Big),
		\end{eqnarray}
		where $m={l,r}$, with $\nu_l = 1$ and $\nu_r = -1$. One can verify that Eq. (\ref{eq: von Neumann equation in universal}), Eq. (\ref{eq: exact final state in universal}), and Eq. (\ref{eq: A7}) are essentially the same as Eq. (A1), Eq. (A4) and Eq. (A7) in  Ref. \cite{nathan2020universal}.
		
		Then, by simultaneously taking the trace norm on both sides of Eq. (\ref{eq: A7}), following the same derivation as in Ref. \cite{nathan2020universal}, we obtain
		\begin{eqnarray}
			\norm{\frac{d}{dt}\rho(t)}_1 \leq \sum_{m,n;\alpha,\beta}\int_0^t ds \abs{C^{mn}_{\alpha,\beta}(t-s)}k^n_\beta(t,s),
		\end{eqnarray}
		where $C^{mn}_{\alpha,\beta}(t-s) = \Tr_B(\hat{B}^m_\alpha(t)\hat{B}^n_\beta(s)\rho_B)$ and $k^n_\beta(t,s) = \norm{\Tr_B\Big(\hat{\mathcal{U}}(t,s)\hat{A}^n_\beta \hat{\mathcal{U}}(s,0)\rho(0)\otimes\rho_B \Big)}_1$. Except for replacing the spectral norm with the trace norm, this equation is identical to Eq. (A13) in Ref. \cite{nathan2020universal}. Using the inequality (\ref{eq: partial trace norm}) and (\ref{eq: norm submultiplicative}), one can verify that $k^n_\beta(t,s) \leq 1,$ which is the same result as Eq. (A18) in Ref. \cite{nathan2020universal}. We then obtain
		\begin{eqnarray}
			\norm{\frac{d}{dt}\rho(t)}_1 \leq 4 \sum_{\alpha\beta}\int_0^t ds \abs{C_{\alpha\beta}(t-s)} \leq \frac{1}{\tau_R}.
		\end{eqnarray}
		This result is the same as in Ref. \cite{nathan2020universal}, except for replacing the spectral norm with the trace norm.
		
		Following the same derivation progress, we may bound the error induced by the Born approximation and the Markov approximation with the trace norm. With the same notation, we have:
		\begin{eqnarray}
			\norm{\xi_B(t)}_1,\  \norm{\xi_M(t)}_1 \leq \frac{\tau_B}{\tau_R^2}.
		\end{eqnarray}
		
		We have provided the error bounds for the Born-Markov approximate master equation and the exact evolution equation. Now, we will present the error of the final states evolved by the two equations. To achieve this, we use Lemma 1 in Ref. \cite{mozgunov2020completely}.
		\begin{lemma}(Mozgunov and Lidar \cite{mozgunov2020completely})
			Assume that
			\begin{eqnarray}
				\dot{x}(t) = \mathcal{L}(x(t)) + \varepsilon, \quad \dot{y}(t) = \mathcal{L}(y(t)), \quad x(0) = y(0),
			\end{eqnarray}
			where $\mathcal{L}$ is a superoperator and $\Lambda$ is a positive constant such that $\sup_{\tau,x:\norm{x}_1=1}\norm{\mathcal{L}_\tau(x)}_1 \leq \Lambda$. Or, assume that 
			\begin{eqnarray}
				\dot{x}(t) = \int_0^t K_{t-\tau}(x(\tau))d\tau + \varepsilon, \quad \dot{y}(t) = \int_0^t K_{t-\tau}(y(\tau)), \quad x(0) = y(0),
			\end{eqnarray}
			where $K_{t-\tau}(x)$is a linear superoperator and $\Lambda$ is a positive constant such that $\sup_{t,x(\tau):\norm{x(\tau)}_1=1}\int_0^t\norm{K_{t-\tau}(x(\tau))}_1 d\tau \leq \Lambda$. Then:
			\begin{eqnarray}
				\forall t \leq \frac{c}{\Lambda}: \norm{x(t) - y(t)}_1 \leq (e^c - 1)\frac{\norm{\varepsilon}_1}{\Lambda}.
			\end{eqnarray}
		\end{lemma}
		
		Using this lemma, by taking $\norm{\varepsilon}_1 = \frac{2\tau_B}{\tau_R^2}$, $\Lambda = \frac{1}{\tau_R}$ and $c = \Lambda t$, we obtain
		
		\begin{eqnarray}
			\label{eq: norm between BM and exact}
			\norm{\rho_{BM}(t) - \rho(t)}_1 \leq 2(e^{t/\tau_R} - 1)\frac{\tau_B}{\tau_R}.
		\end{eqnarray}
		
			Here, $\rho_{BM}(t)$ denotes the system density matrix obtained from the master equation derived through the Born-Markov approximation, written as:
			\begin{eqnarray}
				\label{eq: master equation with BM approximation}
				\frac{d}{dt}\rho_{BM}(t) &=& \sum_{\alpha,\beta}\int_0^t ds C_{\alpha\beta}(s) \left[A_\beta(t-s)\rho_{BM}(t) A_\alpha(t)\right. \notag \\
				&&\left.- A_\alpha(t)A_\beta(t-s)\rho_{BM}(t)\right] + h.c.
			\end{eqnarray}

		\subsubsection{Redfield limits of integration error}
		In addition to the error caused by the Born-Markov approximation, another error needs to be analyzed: the error due to changing the integral $\int^t \to \int^\infty$ in Eq. (\ref{eq: Born-Markov approximation master equation}). For $\xi_l(t)$, using inequality (\ref{eq: norm submultiplicative}), we have
		\begin{eqnarray}
			\norm{\xi_l(t)} \leq 4 \norm{\rho_{BM}(t)}_1 \sum_{\alpha,\beta}\int_t^\infty \abs{C_{\alpha,\beta}(s)} ds 
			\leq 4 \frac{\norm{\rho_{BM}(t)}_1}{t}\sum_{\alpha,\beta}\int_t^\infty s\abs{C_{\alpha,\beta}(s)} ds 
			\leq \frac{\norm{\rho_{BM}(t)}_1\tau_B}{t\tau_S}.
		\end{eqnarray}
		We augment this inequality by a trivial bound:
		\begin{eqnarray}
			\norm{\xi_l(t)} \leq 4 \norm{\rho_{BM}(t)}_1 \sum_{\alpha,\beta}\int_t^\infty \abs{C_{\alpha,\beta}(s)} ds 
			\leq \frac{\norm{\rho_{BM}(t)}_1}{\tau_S}\min(1,\frac{\tau_B}{t}).
		\end{eqnarray}
		We follow the derivation progress in Sec. 6.6 in Ref. \cite{mozgunov2020completely} and obtain:
		\begin{eqnarray}
			\label{eq: norm between BM and Redfield}
			&&\norm{\rho_{BM}(t) - \rho_R(t)}_1 \leq \left\{\begin{array}{ll}
				\norm{\rho_{BM}(t)}_1 \Big(e^{t/\tau_R}-1\Big), & t \leq \tau_B, \\
				\norm{\rho_{BM}(t)}_1 e^{t/\tau_R}\Big(\big(1-e^{-\tau_B/\tau_R}\big)+\frac{\tau_B}{\tau_R}\Big(\ln{\frac{t}{\tau_B}}+\sum_{n=1}^\infty\frac{(-1)^n }{n n!\tau_R^n}\big(t^n-\tau_B^n\big)\Big)\Big), & t > \tau_B.
			\end{array}\right. 
		\end{eqnarray}

		For $t\leq \tau_R\ln{\Big(1+\frac{\tau_R}{2\tau_B}\Big)}$, one can verify $\norm{\rho_{BM}(t)-\rho(t)}_1 \leq 1$ according to Eq. (\ref{eq: norm between BM and exact}). Note that $\norm{\rho(t)}_1$ always be identity since it is a density matrix, we obtain $\norm{\rho_{BM}(t)}_1 \leq 2$. One can verify that $\sum_{n=1}^\infty\frac{(-1)^n }{n n!\tau_R^n}\big(t^n-\tau_B^n\big)\leq 0$ for $t > \tau_B$. By combining these two inequalities with Eq. (\ref{eq: norm between BM and exact}) and Eq. (\ref{eq: norm between BM and Redfield}), we can prove Lemma \ref{lem: Redfield error}.
		
		Using inequality $1-e^{-x} < x$ and $\ln{x} < 2\sqrt{x} - 2$, we can obtain the following corollary.

		\begin{corollary}
			\label{corollary: Redfield error}
			When $\tau_B < t \leq  \tau_R$, the distance between the exact system state $\rho(t)$ and the state derived by the Redfield equation $\rho_R(t)$ can be bounded by
			\begin{eqnarray}
				\norm{\rho_R(t) - \rho(t)}_1 \leq 4e\frac{\sqrt{\tau_B}\sqrt{t}}{\tau_R}.
			\end{eqnarray}
		\end{corollary}
		
		\subsection{Spectral decomposition, rotating-wave approximation and Lindblad equation}
		\label{appsub: spectral decomposition}
		Let $H = \sum_j E_j \Pi_j$ be the spectral decomposition of the system Hamiltonian, where $E_j$ are eigenvalues of $H$, and $\Pi_j$ is the projection onto the subspace spanned by eigenvectors of $E_j$. The set of energy differences is $\Omega \equiv \{E_i-E_j\}$. Then, the generator of the Redfield equation can be written as
		\begin{eqnarray}
			\label{eq: Redfield generator}
			\mathcal{K}(T,J;t)\bullet = \sum_{\omega,\omega'\in\Omega}\sum_{\alpha,\beta} e^{i(\omega'-\omega)t} \Gamma_{\alpha,\beta}(T,J;\omega) 
			\Big(A_\beta(\omega)\bullet A_\alpha^\dag(\omega') - A_\alpha^\dag(\omega') A_\beta(\omega)\bullet\Big) + H.c.,
		\end{eqnarray}
		where $\Gamma_{\alpha,\beta}(T,J;\omega) = \int_{0}^{\infty}ds e^{i\omega s}C_{\alpha,\beta}(T,J;s)$ and $A_\alpha(\omega) = \sum_{m,n\vert E_n-E_m=\omega} \Pi_m A\Pi_n.$ Notice that $A_\alpha(\omega)$ operators satisfy $[H,A_\alpha(\omega)] = -\omega A_\alpha(\omega)$ and $A_\alpha = \sum_{\omega\in\Omega} A_\alpha(\omega)$. 
		
		It is convenient to express $\Gamma_{\alpha,\beta}(T,J;\omega)$ in the form 
		\begin{eqnarray}
			\Gamma_{\alpha,\beta}(T,J;\omega) = \frac{1}{2} \gamma_{\alpha,\beta}(T,J;\omega) + iS_{\alpha,\beta}(T,J;\omega),
		\end{eqnarray}
		where 
		\begin{eqnarray}
			\gamma_{\alpha,\beta}(T,J;\omega) &=& \Gamma_{\alpha,\beta}(T,J;\omega) + \Gamma^*_{\beta,\alpha}(T,J;\omega) = \int_{-\infty}^{\infty}ds e^{i\omega s}C_{\alpha,\beta}(T,J;s), \\
			S_{\alpha,\beta}(T,J;\omega) &=& \frac{1}{2i}\big(\Gamma_{\alpha,\beta}(T,J;\omega) - \Gamma_{\beta,\alpha}^*(T,J;\omega)\big).
		\end{eqnarray}
		For all $\omega$, $\gamma(\omega)$ is a positive semi-definite matrix, and $S(\omega)$ is a Hermitian matrix. With correlation function expression Eq. (\ref{eq:correlationGB}), one can verify that 
		
		\begin{eqnarray}
			\gamma_{\alpha,\beta}(T,J;\omega) = \left\{\begin{array}{ll}
				2\pi J_{\alpha,\beta}^*(\omega)\frac{e^{\omega\slash T}}{e^{\omega/T} - 1}, & \omega > 0, \\
				2\pi J_{\alpha,\beta}(-\omega)\frac{1}{e^{-\omega/T} - 1}, & \omega< 0,
			\end{array}\right. 
		\end{eqnarray}
		
		and 
		\begin{eqnarray}
			\label{eq: dissipation reason}
			\gamma_{\alpha,\beta}(T,J; -\omega) = e^{-\frac{\omega}{T}} \gamma_{\beta,\alpha}(T,J;\omega).
		\end{eqnarray}
		for positive $\omega$.
		
		Under the condition that $1/\tau_R(T,J)$ is small compared to differences between eigenenergies of $H$, we can apply the rotating-wave approximation to the Redfield equation \cite{breuer2002theory}. The rotating wave approximation neglects the terms in Eq. (\ref{eq: Redfield generator}) where 
		$\omega$ and $\omega'$ are not equal (i.e., $\omega \neq \omega'$), thereby yielding the Lindblad equation. The Lindblad equation reads
		\begin{eqnarray}
			\label{eq: Lindblad}
			\frac{d}{dt}\rho_L(t) &=& -i[H+H_{LS}(T,J),\rho_L(t)] + \mathcal{L}(T,J) \rho_L(t)
		\end{eqnarray}
		in the Schrödinger picture, where $H_{LS}(T,J) = \sum_\omega\sum_{\alpha,\beta}S_{\alpha,\beta}(T,J;\omega)A_\alpha^\dag(\omega)A_\beta(\omega)$ is the Lamb shift Hamiltonian satisfying $[H,H_{LS}] = 0$, and
		\begin{eqnarray}
			\mathcal{L}(T,J)\bullet = \frac{1}{2}\sum_{\omega\in\Omega}\sum_{\alpha,\beta} \gamma_{\alpha,\beta}(T,J;\omega) \Big(A_\beta(\omega)\bullet A_\alpha^\dag(\omega) - A_\alpha^\dag(\omega) A_\beta(\omega)\bullet\Big) + H.c..
		\end{eqnarray}
		
		Similar as is mentioned in Sec. \ref{sec: robustness} in the main text, according to Eq. (\ref{eq: Lindblad}), the energy dissipation rate is
		\begin{eqnarray}
			-\frac{d}{dt}\Tr(H\rho_L(t)) = - \Tr(H \mathcal{L}(T,J) \rho_L(t) )
			= \sum_{\omega\in\Omega}\omega \sum_{\alpha,\beta} \gamma_{\alpha,\beta}(T,J;\omega)
			\Tr\Big(A_\beta(\omega)\rho_L(t) A_\alpha^\dag(\omega)\Big),
		\end{eqnarray}
		where we use $[H,H_{LS}] = 0$ and $[H,A_\alpha(\omega)] = -\omega A_\alpha(\omega)$. Since $\{\gamma_{\alpha,\beta}(T,J;\omega)\}$ is a positive is a positive semi-definite matrix, one can verify that $\sum_{\alpha,\beta} \gamma_{\alpha,\beta}(T,J;\omega)\Tr\Big(A_\beta(\omega)\rho_L(t) A_\alpha^\dag(\omega)\Big)$ is always non-negative by diagonalizing this matrix (See Appendix \ref{app: dissipation} for more details). This result suggests that the effect of any bath described by the spectral density $J(\omega)$ on the system can be divided into two parts based on the sign of $\omega$: the part where $\omega>0$ represents the bath reducing the system's energy, while the part where $\omega<0$ represents the bath increasing the system's energy. We provide an upper bound for the latter part in Appendix \ref{app: dissipation}.

		\section{Quasi-steady states and Gaussian stabilization}
		\label{app: steady state}
		In the following two sections, we introduce Quasi-steady states and their energy dissipation under a specific form of the Redfield superoperator, in preparation for proving Lemma \ref{lem:super} in Appendix \ref{app: proof of lemma1}. As mentioned in Sec. \ref{sec: robustness}, if the super bath can be decomposed into a good bath and a complementary bath, the Lindblad equation can also be decomposed into corresponding terms. Furthermore, at sufficiently low temperatures, the term corresponding to the complementary bath will not increase the system's energy under the conditions described in the maintext. However, in open systems where the rotating wave approximation does not hold (e.g., many-body systems), the Lindblad equation may fail. In these two sections, we present concepts and methods designed to address the challenges arising from the inability to apply the rotating-wave approximation. Ultimately, we show that for the Redfield equation, the lower the temperature, the smaller the upper bound on the energy increase of the system due to the term corresponding to the complementary bath.

		In this section, we introduce quasi-steady states and its preparation. We start by defining the Gaussian energy filter.
		
		\begin{definition}
			{\bf Gaussian energy filter.} A Gaussian energy filter $G_\sigma(x)$ with positive real parameter $\sigma$ (i.e., $\sigma > 0$) and real parameter $x$ is a projection 
			\begin{eqnarray}
				G_\sigma(x) &\equiv& \sum_j \sqrt{\frac{\sqrt{2} \sigma}{\sqrt{\pi}}} e^{-\sigma^2 (E_j - x)^2} \Pi_j.
			\end{eqnarray}
			
		\end{definition}
		
		With the Gaussian energy filter, the Gaussian stabilization defined in Definition \ref{def:GS} can be expressed as
		\begin{eqnarray}
			\mathcal{G}_\sigma \rho &=& \int_{-\infty}^{+\infty} dx \, G_\sigma(x) \rho G_\sigma(x).
		\end{eqnarray}
		For an arbitrary state $\rho \in \mathbb{S}$, the operation $\mathcal{G}_\sigma$ reduces off-diagonal elements of $\rho$ in the following way: 
		\begin{eqnarray}
			\mathcal{G}_\sigma \rho &=& \sum_{i,j} e^{-\frac{\sigma^2 (E_i - E_j)^2}{2}} \Pi_i \rho \Pi_j.
		\end{eqnarray}

		Similarly, we define the rectangular energy filter and  $\delta$-stationary states.
		\begin{definition}
			{\bf Rectangular energy filter.} A rectangular energy filter with positive real parameter $\delta$ (i.e.,$\delta > 0$) and real parameter $x$ is a projection 
			\begin{eqnarray}
				F_\delta(x) = \sum_j f_{\delta,j}(x) \Pi_j,
			\end{eqnarray}
			where 
			\begin{eqnarray}
				f_{\delta,j}(x) = 
				\begin{cases} 
					1, & x-\delta \leq E_j \leq x+\delta, \\
					0, & \text{otherwise}.
				\end{cases}
			\end{eqnarray}
		\end{definition}

		\begin{definition}
			{\bf $\delta$-stationary state.} A state $\rho_s \in \mathbb{S}$ is said to be a $\delta$-stationary state if and only if there exists a probability density function $p(x)$ and operator-valued function $\rho'(x)$ such that: (i) $\rho'(x) \in \mathbb{S}$ for all $x$, (ii) $\rho'(x) = F_\delta(x) \rho'(x) F_\delta(x)$ for all $x$, and (iii) 
			\begin{eqnarray}
				\rho_s = \int_{-\|H\|_\infty - \delta}^{\|H\|_\infty + \delta} dx \, p(x) \rho'(x).
				\label{eq: steady}
			\end{eqnarray}
			Here, $\int_{-\|H\|_\infty - \delta}^{\|H\|_\infty + \delta} dx \, p(x) = 1$. 
		\end{definition}

		The $\delta$-stationary state is also referred to as the quasi-steady state, as it shares similar properties with the steady state, which will be demonstrated in Appendix \ref{app: dissipation}. With these definitions, we introduce how to prepare quasi-steady states through Gaussian stabilization by the following lemma:
		
		\begin{lemma}
			\label{lem: delta steady state}
			For all states $\rho \in \mathbb{S}$, there exists a $\delta$-stationary state $\rho_s \in \mathbb{S}$ such that 
			\begin{eqnarray}
				\|\mathcal{G}_\sigma \rho - \rho_s\|_1  \leq \frac{8\sqrt{2} \sigma (\|H\|_\infty + \delta)}{\sqrt{\pi}} e^{-\sigma^2 \delta^2} + 2 e^{-2\sigma^2 \delta^2}.
			\end{eqnarray}
		\end{lemma}
		
		\begin{proof}
			First, we construct the $\delta$-stationary state $\rho_s$. We take 
			\begin{eqnarray}
				q(x) = \Tr[F_\delta(x) G_\sigma(x) \rho G_\sigma(x) F_\delta(x)],
			\end{eqnarray}
			\begin{eqnarray}
				p(x) = \frac{q(x)}{\int_{-\|H\|_\infty - \delta}^{\|H\|_\infty + \delta} dx' \, q(x')}
			\end{eqnarray}
			and 
			\begin{eqnarray}
				\rho'(x) &=& \frac{F_\delta(x) G_\sigma(x) \rho G_\sigma(x) F_\delta(x)}{q(x)}.
			\end{eqnarray}
			Then, we have a $\delta$-stationary state $\rho_s$ according to Eq. (\ref{eq: steady}). 
			
			Second, we introduce the truncated Gaussian quasi-steady state 
			\begin{eqnarray}
				\mathcal{G}^t_\sigma \rho &=& \int_{-\|H\|_\infty - \delta}^{\|H\|_\infty + \delta} dx \, G_\sigma(x) \rho G_\sigma(x).
			\end{eqnarray}
			When $|x| \geq \|H\|_\infty$, 
			\begin{eqnarray}
				\|G_\sigma(x) \rho G_\sigma(x)\|_1 \leq \sqrt{\frac{\sqrt{2} \sigma}{\sqrt{\pi}}} e^{-2\sigma^2 (|x| - \|H\|_\infty)^2}.
			\end{eqnarray}
			Therefore, 
			\begin{eqnarray}\label{eq:error1}
				\|\mathcal{G}_\sigma \rho - \mathcal{G}^t_\sigma \rho\|_1 &\leq& \mathrm{erfc}(\sqrt{2} \sigma \delta) \leq e^{-2\sigma^2 \delta^2}.
			\end{eqnarray}
			
			We have 
			\begin{eqnarray}
				\|F_\delta(x) G_\sigma(x) - G_\sigma(x)\|_\infty \leq \sqrt{\frac{\sqrt{2} \sigma}{\sqrt{\pi}}} e^{-\sigma^2 \delta^2}.
			\end{eqnarray}
			Using the inequality (\ref{eq: norm submultiplicative}), together with $\|\rho\|_1 = 1$, $\|F_\delta(x)\|_\infty = 1$ and $\|G_\sigma(x)\|_\infty \leq \sqrt{\frac{\sqrt{2} \sigma}{\sqrt{\pi}}}$, we have 
			\begin{eqnarray}
				\|F_\delta(x) G_\sigma(x) \rho G_\sigma(x) F_\delta(x) - G_\sigma(x) \rho G_\sigma(x)\|_1 \leq \frac{2\sqrt{2} \sigma}{\sqrt{\pi}} e^{-\sigma^2 \delta^2}.
			\end{eqnarray}
			Let the unnormalized $\delta$-stationary state be 
			\begin{eqnarray}
				\rho^u_s = \int_{-\|H\|_\infty - \delta}^{\|H\|_\infty + \delta} dx \, q(x) \rho'(x).
			\end{eqnarray}
			Then, 
			\begin{eqnarray}\label{eq:error2}
				\|\mathcal{G}^t_\sigma \rho - \rho^u_s\|_1 \leq \frac{4\sqrt{2} \sigma (\|H\|_\infty + \delta)}{\sqrt{\pi}} e^{-\sigma^2 \delta^2}.
			\end{eqnarray}
			
			Similar to the inequality (\ref{eq: partial trace norm}), we have $|\Tr(A)| \leq \|A\|_1$, and we obtain 
			\begin{eqnarray}
				|\Tr(\mathcal{G}_\sigma \rho) - \Tr(\mathcal{G}^t_\sigma \rho)| \leq e^{-2\sigma^2 \delta^2}
			\end{eqnarray}
			and 
			\begin{eqnarray}
				|\Tr(\mathcal{G}^t_\sigma \rho) - \Tr(\rho^u_s)| \leq \frac{4\sqrt{2} \sigma (\|H\|_\infty + \delta)}{\sqrt{\pi}} e^{-\sigma^2 \delta^2}.
			\end{eqnarray}
			Then, 
			\begin{eqnarray}
				|\Tr(\rho^u_s) - 1| \leq \frac{4\sqrt{2} \sigma (\|H\|_\infty + \delta)}{\sqrt{\pi}} e^{-\sigma^2 \delta^2} + e^{-2\sigma^2 \delta^2}.
			\end{eqnarray}
			Notice that 
			\begin{eqnarray}
				\rho_s = \frac{\rho^u_s}{\Tr(\rho^u_s)}.
			\end{eqnarray}
			Therefore, 
			\begin{eqnarray}
				\label{eq:error3}
				\|\rho^u_s - \rho_s\|_1 \leq \left|\Tr(\rho^u_s) - 1\right| \|\rho_s\|_1 \leq \frac{4\sqrt{2} \sigma (\|H\|_\infty + \delta)}{\sqrt{\pi}} e^{-\sigma^2 \delta^2} + e^{-2\sigma^2 \delta^2}.
			\end{eqnarray}
			
			Finally, with Eqs.~(\ref{eq:error1}), (\ref{eq:error2}), and (\ref{eq:error3}), we have
			\begin{eqnarray}
				\|\mathcal{G}_\sigma \rho - \rho_s\|_1 \leq \|\mathcal{G}_\sigma \rho - \mathcal{G}^t_\sigma \rho\|_1 + \|\mathcal{G}^t_\sigma \rho - \rho^u_s\|_1 + \|\rho^u_s - \rho_s\|_1 
				\leq \frac{8\sqrt{2} \sigma (\|H\|_\infty + \delta)}{\sqrt{\pi}} e^{-\sigma^2 \delta^2} + 2 e^{-2\sigma^2 \delta^2}.
			\end{eqnarray}
		\end{proof}

		\section{Dissipation of quasi-steady states}
		\label{app: dissipation}

		To overcome the difficulties that may arise when the rotating wave approximation fails, we also introduce several complementary concepts in addition to introducing the steady state. The main conclusions of this section are summarized in Lemma \ref{lem: cg energy transfer} and Lemma \ref{lem: energy increase} in \ref{appsub: main result}.

		\subsection{Coarse-grained Hamiltonian}
		
		The coarse-grained Hamiltonian with parameters $(\delta,x)$ reads 
		\begin{eqnarray}
			H_\delta(x) = \sum_{j=-\infty}^{\infty} (x+2j\delta)F_\delta(x+2j\delta)
		\end{eqnarray}
		Then, 
		\begin{eqnarray}
			\norm{H_\delta(x)-H}_\infty \leq \delta
		\end{eqnarray}
		and 
		\begin{eqnarray}
			\norm{e^{-iH_\delta(x)t} - e^{-iHt}}_\infty \leq \delta\abs{t},
		\end{eqnarray}
		where we use $\abs{ e^{-i\delta t} - 1 } = 2\abs{ \sin\frac{\delta t}{2} } \leq \delta \abs{t}$. 
		
		Let $A_{\alpha,\delta}(t,x) = e^{iH_\delta(x)t}A_\alpha e^{-iH_\delta(x)t}$ be time-dependent system operators according to the coarse-grained Hamiltonian. Then, 
		\begin{eqnarray}
			\norm{A_{\alpha,\delta}(t,x) - A_\alpha(t)}_\infty \leq 2\delta \abs{t}. 
		\end{eqnarray}
		For a given set of correlation functions $C_{\alpha,\beta}(s)$, the Redfield equation with the coarse-grained Hamiltonian is given by 
		\begin{eqnarray}
			\mathcal{K}_{R,\delta}(t,x)\bullet = \sum_{\alpha,\beta}\int_0^\infty ds C_{\alpha\beta}(s) \big(A_{\beta,\delta}(t-s,x)\bullet A_{\alpha,\delta}(t,x) - A_{\alpha,\delta}(t,x)A_{\beta,\delta}(t-s,x)\bullet\big) + H.c.,
		\end{eqnarray}
		and we have 
		\begin{eqnarray}
			\norm{\mathcal{K}_{R,\delta}(t,x) - \mathcal{K}(t)}_1 
			\leq \sum_{\alpha,\beta}\int_0^\infty ds \delta(16\abs{t}+8\abs{s})C_{\alpha\beta}(s) 
			\leq 2\delta\left(\frac{2\abs{t}}{\tau_R} + \frac{\tau_B}{\tau_R}\right).
		\end{eqnarray}
		
		\subsection{Spectral decomposition}
		
		Similarly, let $\Omega_\delta \equiv \{2j\delta \vert j \in \mathbb{Z} \}$ and $\omega,\omega' \in \Omega_\delta$, follow the spectral decomposition method mentioned in Appendix \ref{appsub: spectral decomposition}, we have
		\begin{eqnarray}
			\mathcal{K}_{R,\delta}(t,x)\bullet = \sum_{\omega,\omega'\in\Omega_\delta}\sum_{\alpha,\beta} e^{i(\omega'-\omega)t} \Gamma_{\alpha,\beta}(\omega) 
			\big(A_{\beta,\delta}(\omega,x)\bullet A_{\alpha,\delta}^\dag(\omega',x) - A_{\alpha,\delta}^\dag(\omega',x) A_{\beta,\delta}(\omega,x)\bullet\big) + H.c.,
		\end{eqnarray}
		where $A_{\alpha,\delta}(\omega,x)$ are operators satisfying $[H_\delta(x),A_{\alpha,\delta}(\omega,x)] = -\omega A_{\alpha,\delta}(\omega,x)$ and 
		\begin{eqnarray}
			A_\alpha = \sum_{\omega\in\Omega_\delta} A_{\alpha,\delta}(\omega,x).
		\end{eqnarray}
		
		\subsection{Dissipation of $\delta$-stationary states}
		\label{appsub: main result}
		We denote the total energy transferred from the system to the bath at the time t according to the first-order contribution of the Redfield equation as
		\begin{eqnarray}
			D_R(t,\rho_s)\equiv -\int_0^t ds \Tr\big(H\mathcal{K}(s)\rho_s\big)
		\end{eqnarray}
		when $\rho_s$ is the initial state.
		Substitute Eq. (\ref{eq: steady}) into $D_R(t,\rho_s)$, we obtain 
		\begin{eqnarray}
			D_R(t,\rho_s) = -\int_{-\norm{H}_\infty-\delta}^{\norm{H}_\infty+\delta} dx p(x) \Tr\big(H\mathcal{K}(s)\rho'(x)\big) = -\int_{-\norm{H}_\infty-\delta}^{\norm{H}_\infty+\delta} dx p(x) D_R(t,\rho'(x))
		\end{eqnarray}
		Similarly, we define the coarse-grained energy transfer 
		\begin{eqnarray}
			D_{R,\delta}(t,p,\rho') \equiv \int_{-\norm{H}_\infty-\delta}^{\norm{H}_\infty+\delta} dx \, p(x) D'_{R,\delta}(t,\rho',x),
		\end{eqnarray}
		where 
		\begin{eqnarray}
			D'_{R,\delta}(t,\rho',x) \equiv -\int_0^t ds \Tr\big(H_\delta(x)\mathcal{K}_{R,\delta}(s,x)\rho'(x)\big).
		\end{eqnarray}
		One then can verify the following lemma by using the triangle inequality.
		\begin{lemma}
			\label{lem: cg energy transfer}
			For any $\delta-$steady initial state $\rho_s$ in the form of Eq. (\ref{eq: steady}), the total energy transferred from the system to the bath at the time $t$  according to the first-order contribution of the Redfield equation, i.e., $D_R(t,\rho_s)$, can be approximated by the corresponding coarse-grained energy transfer, i.e., $D'_{R,\delta}(t,\rho',x)$, and the inequality
			\begin{eqnarray}
				\label{eq: energy transfer error 1}
				\abs{D_{R,\delta}(t,p,\rho') - D_R(t,\rho_s)} 
				\leq \delta \Big(\frac{t}{\tau_R} + 2\Big(\frac{t}{\tau_R} + \frac{\tau_B}{\tau_R}\Big)\norm{H}_\infty t\Big).
			\end{eqnarray}
			holds.
		\end{lemma}

		Using the spectral decomposition, the coarse-grained energy transfer $D'_{R,\delta}(t,\rho',x)$ can be rewritten as
		\begin{eqnarray}
			\label{eq: SD of cg energy transfer}
			D'_{R,\delta}(t,\rho',x) =  \int_0^t ds \sum_{\omega,\omega'\in\Omega_\delta}\sum_{\alpha,\beta} e^{i(\omega'-\omega)s} 
			\big(X_{\alpha,\beta}(\omega',\omega)+iY_{\alpha,\beta}(\omega',\omega)\big) 
			Z_{\beta,\alpha,\delta}(\omega,\omega',x),
		\end{eqnarray}
		where 
		\begin{eqnarray}
			X_{\alpha,\beta}(\omega',\omega) = \frac{1}{2}\big(\omega'\gamma_{\alpha,\beta}(\omega) + \omega\gamma_{\alpha,\beta}(\omega')\big), \quad Y_{\alpha,\beta}(\omega',\omega) = \omega' S_{\alpha,\beta}(\omega) - \omega S_{\alpha,\beta}(\omega'),
		\end{eqnarray}
		and 
		\begin{eqnarray}
			Z_{\beta,\alpha,\delta}(\omega,\omega',x) 
			= \Tr\Big(A_{\beta,\delta}(\omega,x) \rho'(x) A_{\alpha,\delta}^\dag(\omega',x)\Big) 
			= \delta_{\omega,\omega'}\Tr\Big(A_{\beta,\delta}(\omega,x) \rho'(x) A_{\alpha,\delta}^\dag(\omega,x)\Big).
		\end{eqnarray}
		Then, 
		\begin{eqnarray}
			D_{R,\delta}(t,p,\rho') &=& \int_{-\norm{H}_\infty-\delta}^{\norm{H}_\infty+\delta} dx \, p(x)\, t \sum_{\omega\in\Omega_\delta}\sum_{\alpha,\beta} \omega\gamma_{\alpha,\beta}(\omega) 
			\Tr\Big(A_{\beta,\delta}(\omega,x) \rho'(x) A_{\alpha,\delta}^\dag(\omega,x)\Big) \notag \\
			&=& D'_{R,\delta,+}(t,\rho',x) + D'_{R,\delta,-}(t,\rho',x),
		\end{eqnarray}
		where
		\begin{eqnarray}
			D_{R,\delta,+}(t,p,\rho') &\equiv& \int_{-\norm{H}_\infty-\delta}^{\norm{H}_\infty+\delta} dx \, p(x)\, t \sum_{\substack{\omega\in\Omega_\delta \\ \omega>0}}\omega\sum_{\alpha,\beta} \gamma_{\alpha,\beta}(\omega) 
			\Tr\Big(A_{\beta,\delta}(\omega,x) \rho'(x) A_{\alpha,\delta}^\dag(\omega,x)\Big), \\
			D_{R,\delta,-}(t,p,\rho') &\equiv& \int_{-\norm{H}_\infty-\delta}^{\norm{H}_\infty+\delta} dx \, p(x)\, t \sum_{\substack{\omega\in\Omega_\delta \\ \omega<0}}\omega\sum_{\alpha,\beta} \gamma_{\alpha,\beta}(\omega) 
			\Tr\Big(A_{\beta,\delta}(\omega,x) \rho'(x) A_{\alpha,\delta}^\dag(\omega,x)\Big)
		\end{eqnarray}

		One can find that $D'_{R,\delta}(t,\rho',x)$ is the energy transfer according to the Lindblad equation after rotating wave approximation. Similar to Appendix \ref{appsub: spectral decomposition}, $D'_{R,\delta}(t,\rho',x)$ can be divided into two terms, i.e., $D'_{R,\delta,+}(t,\rho',x)$ and $D'_{R,\delta,-}(t,\rho',x)$. The former represents energy dissipation from the system, while the latter represents energy increase in the system.
		
		We then are going to find a lower bound of $D_{R,\delta,-}(t,p,\rho')$. Since the matrix $\{ \gamma_{\alpha,\beta}(\omega)\}$ is semi-positive, we can diagonalize it with a unitary matrix. We have $\gamma(\omega) = V(\omega)\Lambda(\omega) V^\dagger(\omega)$ where $V(\omega)$ is an unitary matrix with elements $\{v_{ij}(\omega)\}$ and $\Lambda(\omega)$ is a diagonal matrix with diagonal elements $\{\lambda_i(\omega)\}$. Then we have
		$\gamma_{\alpha,\beta}(\omega) = \sum_{i = 1}^N \lambda_i(\omega) v_{\alpha i}(\omega)v^\ast_{\beta i}(\omega).$ According to the Gershgorin circle theorem, eigenvalue $\lambda_i(\omega)$ satisfy
		$\lambda_i(\omega) \leq \sum_{j} \abs{\gamma_{i,j}(\omega)}.$
		
		Define operator $C_i(\omega) \equiv \frac{1}{\sqrt{N}}\sum_\alpha v^\ast_{\alpha i}(\omega) A_\alpha$. Using $\sum_\alpha \abs{v_{\alpha i}(\omega)}^2 = 1,$ one can prove that $\norm{C_i(\omega)}_\infty \leq 1.$
		
		Thus, the energy transfer $D_{R,\delta,-}(t,p,\rho')$ can be rewritten as
		\begin{eqnarray}
			D_{R,\delta,-}(t,p,\rho') =  \int_{-\norm{H}_\infty-\delta}^{\norm{H}_\infty+\delta} dx \, p(x)\, t \sum_{\substack{\omega\in\Omega_\delta \\ \omega<0}}\omega \sum_{i = 1}^N N\lambda_i(\omega)
			\Tr\Big(\Pi(x-\omega)C_i(\omega) \rho'(x) C^\dagger_i(\omega) \Pi(x-\omega) \Big). 
		\end{eqnarray}
		
		Since $0\leq\Tr\Big(\Pi(x-\omega)C_i(\omega) \rho'(x) C^\dagger_i(\omega) \Pi(x-\omega) \Big)\leq 1$, we have
		\begin{eqnarray}
			D_{R,\delta,-}(t,p,\rho') \geq \frac{N t}{2\tau_R} \omega e^{\omega/T} , \quad \omega < 0,
		\end{eqnarray}
		where we use Eq. (\ref{eq: dissipation reason}). Thus, we have
		\begin{eqnarray}
			\label{eq: energy transfer error 2}
			D_{R,\delta,-}(t,p,\rho') \geq -\frac{N t T}{2e\tau_R}.
		\end{eqnarray}
		
		Combining from Eq. (\ref{eq: SD of cg energy transfer}) to Eq. (\ref{eq: energy transfer error 2}), we can prove the following Lemma \ref{lem: energy increase} using triangle inequality.
		\begin{lemma}
			\label{lem: energy increase}
			A  coarse-grained energy transfer $D'_{R,\delta}(t,\rho',x)$ can be divided to two terms, $D_{R,\delta,+}(t,p,\rho')$ and $D_{R,\delta,-}(t,p,\rho')$. The former decreases the system's energy, while the latter increases the system's energy. The upper bound of the energy increase is described by the inequality (\ref{eq: energy transfer error 2}).
		\end{lemma}

		\section{Proof of Lemma \ref{lem:super} and Theorem \ref{the:complexity}}
		\label{app: proof of lemma1}
		In this section, we provide the proofs of Lemma \ref{lem:super} and Theorem \ref{the:complexity}. We also present the parameter update process, which brings them within their proper value ranges in $O(\log(N,b,r,h))$ times.

		\subsection{Proof of Lemma \ref{lem:super}}

		Lemma \ref{lem:super} provides the lower bound on the energy decrease of the system within time $t$ due to the super bath described by the spectral density $J_{\text{S}}(\omega)$. To prove this, we perform approximations step by step and calculate the error for each approximation.
		
		{\bf The first step,} replace the exact evolution with the Redfield equation evolution.
		For the term $\Tr(H\rho) - \Tr(H\mathcal{M}(H,T,g^2J_{\text{S}};t)\mathcal{G}_\sigma\rho)$ in Eq. (\ref{eq:dissipation_bound}), we use $\Tr(H\rho_R(t))$ to approximate $\Tr(H\mathcal{M}(H,T,g^2J_{\text{S}};t)\mathcal{G}_\sigma\rho)$, where ${\rho}_R(t)$ denotes the state derived from the Redfield equation in the interaction picture when the initial state is $\mathcal{G}_\sigma\rho$. We have
		\begin{eqnarray}
			\label{eq: proof progress 1}
			\Tr(H\rho) - \Tr(H\mathcal{M}(H,T,g^2J_{\text{S}};t)\mathcal{G}_\sigma \rho) 
			\geq \Tr(HG_\sigma\rho) - \Tr(H {\rho}_R(t)) - \abs{\Tr(H {\rho}_R(t))-\Tr(H {\rho}(t)}),
		\end{eqnarray}
		where we use $\Tr(H\rho) = \Tr(H\mathcal{G}_\sigma \rho)$, and ${\rho}(t)$ denotes the state derived from the exact evolution equation in the interaction picture when the initial state is $\mathcal{G}_\sigma\rho$. Note that the picture transformation does not change the energy expectation value.
		
		For the last term of the right-hand side in Eq. (\ref{eq: proof progress 1}), we have
		\begin{eqnarray}
			\abs{\Tr(H {\rho}_R(t))-\Tr(H {\rho}(t))} 
			\leq \norm{H({\rho}_R(t)-{\rho}(t))}_1 
			\leq  \norm{H}_\infty \norm{{\rho}_R(t)-{\rho}(t)}_1,
		\end{eqnarray}
		which can then be bounded using Corollary \ref{corollary: Redfield error} when $\tau_B(T,J_{\text{S}}) \leq t \leq \tau_R(T,J_{\text{S}})/g^2$.
		
		{\bf The second step,} neglect the higher-order terms of the dynamical map corresponding to the Redfield equation.
		For ${\rho}_R(t)$, similar to the Dyson series expansion, we have the integral form
		\begin{eqnarray}
			\rho_R(t) = \mathcal{G}_\sigma \rho + \int_0^t dt_1 \mathcal{K}(T,g^2J_{\text{S}};t_1) \mathcal{G}_\sigma \rho + \int_0^t dt_1 \mathcal{K}(T,g^2J_{\text{S}};t_1) \int_0^{t_1} dt_2 \mathcal{K}(T,g^2J_{\text{S}};t_2) \mathcal{G}_\sigma \rho + ...
		\end{eqnarray}
		Then, $Tr(HG_\sigma\rho) - \Tr(H {\rho}_R(t))$ can be expressed as
		\begin{eqnarray}
			\label{eq: proof progress 2}
			\Tr(HG_\sigma\rho) - \Tr(H {\rho}_R(t)) &=& -\Tr(H\int_0^t dt_1 \mathcal{K}(T,g^2J_{\text{S}};t_1) \mathcal{G}_\sigma \rho) \notag \\
			&-& \Tr(H\int_0^t dt_1 \mathcal{K}(T,g^2J_{\text{S}};t_1) \int_0^{t_1} dt_2 \mathcal{K}(T,g^2J_{\text{S}};t_2) \mathcal{G}_\sigma \rho) - ...
		\end{eqnarray}

		For the nth term in Eq. (\ref{eq: proof progress 2}), one can verify that
		\begin{eqnarray}
			\abs{\Tr(H\int_0^t dt_1 \mathcal{K}(T,g^2J_{\text{S}};t_1)\int_0^{t_1} dt_2 \mathcal{K}(T,g^2J_{\text{S}};t_2)... \int_0^{t_{n-1}}dt_n \mathcal{K}(T,g^2J_{\text{S}};t_n) \mathcal{G}_\sigma \rho)} \leq \norm{H}_\infty \frac{g^{2n}t^n}{n!\tau^n_R(T,J_{\text{S}})},
		\end{eqnarray}
		where we use $\norm{\mathcal{K}(T,g^2J_{\text{S}};t)\bullet}_1 \leq \frac{g^2}{\tau_R(T,J_{\text{S}})}\norm{\bullet}_1$ and inequality (\ref{eq: norm submultiplicative}). Then all higher-order terms in Eq. (\ref{eq: proof progress 2}) are bounded by
		\begin{eqnarray}
			&&\sum_{n=2}^\infty \abs{\Tr(H\int_0^t dt_1 \mathcal{K}(T,g^2J_{\text{S}};t_1)\int_0^{t_1} dt_2 \mathcal{K}(T,g^2J_{\text{S}};t_2)... \int_0^{t_{n-1}}dt_n \mathcal{K}(T,g^2J_{\text{S}};t_n) \mathcal{G}_\sigma \rho)} \notag \\
			&\leq& \norm{H}_\infty (e^{g^2t/\tau_R(T,J_{\text{S}})} - 1 - \frac{g^2t}{\tau_R(T,J_{\text{S}})}).
		\end{eqnarray}

		{\bf The third step,}  replace the initial state $\mathcal{G}_\sigma \rho$ with the corresponding $\delta-$steady state $\rho_s$. We approximate the first term on the right-hand side of Eq. (\ref{eq: proof progress 2}) by $D_R(T,g^2J_{\text{S}};t,\rho_s)= -\Tr(H\int_0^t dt_1 \mathcal{K}(T,g^2J_{\text{S}};t_1) \rho_s)$. We have
		\begin{eqnarray}
			- \Tr(H \int_0^t dt_1 \mathcal{K}(T,g^2J_{\text{S}};t_1) \mathcal{G}_\sigma \rho)  &\geq&  
			D_R(T,g^2J_{\text{S}};t,\rho_s)  - \abs{\Tr(H\int_0^t dt_1 \mathcal{K}(T,g^2J_{\text{S}};t_1)(\mathcal{G}_\sigma \rho-\rho_s))} \notag \\
			&\geq&  D_R(T,g^2J_{\text{S}};t,\rho_s)  - \norm{H}_\infty\frac{g^2t}{\tau_R(T,J_{\text{S}})}\norm{\mathcal{G}_\sigma \rho-\rho_s}_1.
		\end{eqnarray}
		The last term can then be bounded by Lemma \ref{lem: delta steady state}.
		
		{\bf The fourth step,} use the coarse-grained energy transfer $D_{R,\delta}(T,g^2J_{\text{S}}; t,p,\rho')$ to approximate $D_R(T,g^2J_{\text{S}};t,\rho_s)$ where
		\begin{eqnarray}
			D_{R,\delta}(T,g^2J_{\text{S}}; t,p,\rho') = - \int_{-\|H\|_\infty - \delta}^{\|H\|_\infty + \delta} dx \, p(x)\Tr(H_\delta(x) \int_0^t dt_1 \mathcal{K}_\delta(T,g^2J_{\text{S}};x,t_1) \rho'(x)).
		\end{eqnarray}
		According to Lemma \ref{lem: cg energy transfer}, we obtain
		\begin{eqnarray}
			D_R(T,g^2J_{\text{S}};t,\rho_s)
			\geq  D_{R,\delta}(T,g^2J_{\text{S}}; t,p,\rho')
			- g^2t\delta \Big(\frac{1}{\tau_R(T,J_{\text{S}})} + 2\Big(\frac{t}{\tau_R(T,J_{\text{S}})} + \frac{\tau_B(T,J_{\text{S}})}{\tau_R(T,J_{\text{S}})}\Big)\norm{H}_\infty \Big).
		\end{eqnarray}
		
		{\bf The fifth step,} use spectral decomposition to analyse $D_{R,\delta}(T,g^2J_{\text{S}}; t,p,\rho')$.
		According to the discussion in Appendix \ref{app: dissipation}, $D_{R,\delta}(T,g^2J_{\text{S}}; t,p,\rho')$ can be divided into two terms: 
		\begin{eqnarray}
			D_{R,\delta}(T,g^2J_{\text{S}}; t,p,\rho') = D_{R,\delta,+}(T,g^2J_{\text{S}}; t,p,\rho') + D_{R,\delta,-}(T,g^2J_{\text{S}}; t,p,\rho'),
		\end{eqnarray}
		where $D_{R,\delta,+}(T,g^2J_{\text{S}}; t,p,\rho')$ and $D_{R,\delta,-}(T,g^2J_{\text{S}}; t,p,\rho')$ represent the energy decrease and the energy increase of the system due to the super bath, respectively. Since the super bath can be divided into the good bath and the complementary bath, this decomposition implies that the effects of the good bath and the complementary bath on the system can also be divided into two parts: one corresponding to the energy decrease of the system and the other to the energy increase of the system. Thus, $D_{R,\delta,+}(T,g^2J_{\text{S}}; t,p,\rho')$ has a lower bound provided by the good bath. Considering that the good bath constitutes at least $1 / {b}$ of the "composition" of the super bath, we have
		\begin{eqnarray}
			D_{R,\delta,+}(T,g^2J_{\text{S}}; t,p,\rho') &\geq& \frac{1}{{b}}D_{R,\delta,+}(T,g^2J_{\text{g-sub}}; t,p,\rho') \geq \frac{1}{{b}}D_{R,\delta}(T,g^2J_{\text{g-sub}}; t,p,\rho') \notag \\
			&\geq& \frac{1}{{b}}D_{R}(T,g^2J_{\text{g-sub}}; t,\rho_s) - \frac{1}{{b}}\abs{D_{R,\delta}(T,g^2J_{\text{g-sub}}; t,p,\rho') - D_{R}(T,g^2J_{\text{g-sub}}; t,\rho_s)} \notag \\
			&\geq& Pt - \frac{g^2t\delta}{{b}} \Big(\frac{1}{\tau_R(T,J_{\text{g-sub}})} + 2\Big(\frac{t}{\tau_R(T,J_{\text{g-sub}})} + \frac{\tau_B(T,J_{\text{g-sub}})}{\tau_R(T,J_{\text{g-sub}})}\Big)\norm{H}_\infty \Big),
		\end{eqnarray}
		where $P$ is defined in Lemme \ref{lem:super}, and we use Lemma \ref{lem: cg energy transfer} again.
		
		The energy increase of the system due to the super bath, i.e., $D_{R,\delta,-}(T,g^2J_{\text{S}}; t,p,\rho')$, can then be bounded by Lemma \ref{lem: energy increase}. That is, we have
		\begin{eqnarray}
			D_{R,\delta,-}(T,g^2J_{\text{S}}; t,p,\rho') \geq -\frac{g^2 t N T}{2e\tau_R(T,J_{\text{S}})}
		\end{eqnarray}
		
		{\bf Finally,} combining these error terms, together with $\norm{H}_\infty \leq {h}$, we obtain the inequality
		\begin{eqnarray}
			\Tr(H\rho) - \Tr(H\mathcal{M}(H,T,g^2J_{\text{S}};t)\mathcal{G}_\sigma\rho) 
			\geq Pt - \epsilon,
		\end{eqnarray}
		when $\tau_B(T,J_{\text{S}}) \leq t \leq \tau_R(T,J_{\text{S}})/g^2$, where 
		\begin{eqnarray}
			\label{eq: expression of epsilon}
			\epsilon &=&  4e\frac{g^2\sqrt{\tau_B(T,J_{\text{S}})}\sqrt{t}}{\tau_R(T,J_{\text{S}})}{h} +  (e^{g^2t/\tau_R(T,J_{\text{S}})} - 1 - \frac{g^2t}{\tau_R(T,J_{\text{S}})}){h} \notag \\
			&+& g^2 t\delta \Big(\frac{1}{\tau_R(T,J_{\text{S}})} + 2\Big(\frac{t}{\tau_R(T,J_{\text{S}})} + \frac{\tau_B(T,J_{\text{S}})}{\tau_R(T,J_{\text{S}})}\Big){h} \Big) + {h} \frac{g^2t}{\tau_R(T,J_{\text{S}})}\Big(\frac{8\sqrt{2} \sigma ({h} + \delta)}{\sqrt{\pi}} e^{-\sigma^2 \delta^2} + 2 e^{-2\sigma^2 \delta^2}\Big) \notag \\
			&+& \frac{g^2t\delta}{{b}} \Big(\frac{1}{\tau_R(T,J_{\text{g-sub}})} + 2\Big(\frac{t}{\tau_R(T,J_{\text{g-sub}})} + \frac{\tau_B(T,J_{\text{g-sub}})}{\tau_R(T,J_{\text{g-sub}})}\Big){h} \Big) + \frac{g^2t N T}{2e\tau_R(T,J_{\text{S}})}.
		\end{eqnarray}
		
		\subsection{The proper range of parameter values and proof of Theorem \ref{the:complexity}.} 
		
		We then demonstrate that the inequality $\epsilon \leq P t / 2$ holds when the parameters $t$, $g$, $\delta$,  $\sigma$, and $T$ are within an appropriate range. This means that after every finite time interval $t$ of evolution, the system will dissipate at least $P t / 2$ amount of energy. Since the expression for $\epsilon$ contains a total of six terms, a simple idea is to make each term less than $c \equiv P t / 12$, thus ensuring that $\epsilon$ is less than $P t / 2$.

		In the following, we discuss the appropriate ranges of these parameters in four parts. Each part involves one or two parameters and the corresponding terms in the expression for $\epsilon$.

		{\bf First, the range of $T$.}
		
		The last term of the expression for $\epsilon$ is bounded by $c$ when the temperature $T$ is sufficiently low. Since 
		\begin{eqnarray}
			\label{eq: tauRTJS bound}
			\frac{1}{\tau_R(T,J_{\text{S}})} = 4\sum_{\alpha,\beta} \int_0^\infty ds \abs{C_{\alpha,\beta}(T,J_{\text{S}};s)} = 4 N \int_0^\infty ds \abs{C(T,\mathfrak{J};s)} =  \frac{N}{\tau_R(T,\mathfrak{J})} \leq \frac{N}{\tau_{R,m}},
		\end{eqnarray}
		where $\tau_{R,m} = \min\{\tau_R(T,\mathfrak{J}) \vert 
		T\in (0,1]\}$ is a computable positive number when a valid $\mathfrak{J}$ is determined, we let 
		
		\begin{eqnarray}
			\frac{g^2t N^2 T}{2e\tau_{R,m}} \leq \frac{P t}{12},
		\end{eqnarray}
		the last term of the expression for $\epsilon$ is less than $c$. The proper range of temperature $T$ is
		\begin{eqnarray}
			T \leq \frac{e\tau_{R,m}}{6N^2 b r}
		\end{eqnarray}

		{\bf Second, the range of $g$ and $t$.}
		
		Although the time $t$ can take values within an appropriate range to make sure that $\tau_B(T,J_{\text{S}}) \leq t \leq \tau_R(T,J_{\text{S}})/g^2$ holds, for simplicity, we choose:
		\begin{eqnarray}
			\label{eq: time step}
			t = \frac{\tau_R(T,J_{\text{S}})}{g}
		\end{eqnarray}
		in the following, and let $g < \min(1/2, \tau_R(T,J_{\text{S}})/\tau_B(T,J_{\text{S}}))$. Note that $\tau_R(T,J_{\text{S}})$ and $\tau_B(T,J_{\text{S}})$ are calculable when the temperature $T$ is given.
		
		For the second term of the expression for $\epsilon$, when $g^2 t / \tau_R(T, J_{\text{S}}) \leq \frac{1}{2}$, we have
		\begin{eqnarray}
			{h} \Big(e^{g^2t/\tau_R(T,J_{\text{S}})} - 1 - \frac{g^2t}{\tau_R(T,J_{\text{S}})}\Big)  \leq  {h}\frac{g^4 t^2}{\tau_R^2(T,J_{\text{S}})}.
		\end{eqnarray}
		where we use the inequality $e^x - 1 - x \leq x^2$ for $0 \leq x \leq 1$.
		One then can verify that the right-hand side term will be less than $c$ when $g < \tau_R(T, J_{\text{S}}) / 12{b}{r}{h}$.  
		
		The first term of the expression for $\epsilon$ can be bound by $c$ when $g < \tau^{3}_R(T, J_{\text{S}}) / 2304 e^2 {b}^2{r}^2{h}^2\tau_B(T,J_{\text{S}})$. We summarize the inequalities concerning $g$ as follows:
		\begin{eqnarray}
			g &<& \min\Big(\frac{1}{2}, \frac{\tau_{R,m}}{N\tau_B(T,\mathfrak{J})}, \frac{\tau_{R,m}}{12N  {b}{r}{h}}, \frac{\tau_{R,m}^3(\mathfrak{J})}{2304 e^2  N^3{b}^2{r}^2{h}\tau_B(T,\mathfrak{J})}\Big),
		\end{eqnarray}
		where we use
		\begin{eqnarray}
			\frac{\tau_B(T,J_{\text{S}})}{\tau_R(T,J_{\text{S}})} = 4\sum_{\alpha,\beta}\int_0^\infty ds s\abs{C_{\alpha,\beta}(T,J_{\text{S}};s)} = 4N \int_0^\infty ds s \abs{C(T,\mathfrak{J}; s)} = N \frac{\tau_B(T,\mathfrak{J})}{\tau_R(T,\mathfrak{J})},
		\end{eqnarray}
		and Eq. (\ref{eq: tauRTJS bound}). Note that $\tau_B(T,\mathfrak{J})$ is computable when temperature $T$ is determined.

		{\bf Third, the range of $\delta$.}
		
		For the third and fifth terms of the expression for $\epsilon$, as long as $\delta$ is sufficiently small, these terms will be less than $c$. More specifically, let
		\begin{eqnarray}
			\label{eq: delta_1}
			\delta \Big(\frac{N}{\tau_{R,m}} + 2\Big(\frac{1}{g} + N\frac{\tau_B(T,\mathfrak{J})}{\tau_{R,m}}\Big){h} \Big) \leq \frac{1}{12{b}{r}},
		\end{eqnarray}
		and 
		\begin{eqnarray}
			\label{eq: delta_2}
			\delta \Big(\frac{{b}}{\tau_{R,m}} + 2\Big(\frac{{b} \tau_{R,M}}{gN\tau_{R,m}} + \frac{{b}^2 \tau_B(T,\mathfrak{J})}{\tau_{R,m}}\Big){h} \Big) \leq \frac{1}{12{r}},
		\end{eqnarray}
		these terms are less than $c$, where we use Eq. (\ref{eq: 4}), Eq. (\ref{eq: 5}), and Eq. (\ref{eq: tauRTJS bound}). 
		
		{\bf Fourth, the range of $\sigma$.}
		
		The third term of the expression for $\epsilon$ will be less than $c$ as long as $\sigma$ is large enough since it decreases exponentially with the increase of $\sigma$. More specifically, let $\sigma \geq \Big(\frac{\sqrt{\pi}}{4\sqrt{2}}\Big)^{1/3}\frac{1}{\delta}\Big)$ to make sure that $2e^{-2\sigma^2\delta^2} \leq \frac{8\sqrt{2}\delta\sigma}{\sqrt{\pi}}e^{-\sigma^2\delta^2}$, the inequality can be replaced by a loose form
		\begin{eqnarray}
			\frac{g^2t N}{\tau_{R,m}} \frac{16\sqrt{2}\sigma{h}^2}{\sqrt{\pi}}e^{-\sigma^2\delta^2} \leq c
		\end{eqnarray}
		under an evidently valid condition that $2\delta \leq {h}$. Using $e^{-x} < \frac{1}{x}$ when $x>0$, We obtain the range of values for $\sigma$
		\begin{eqnarray}
			\sigma > \max{\Big(\frac{192\sqrt{2}  N  {b}{r}{h}^2}{\sqrt{\pi}\delta^2\tau_{R,m}}, \Big(\frac{\sqrt{\pi}}{4\sqrt{2}}\Big)^{1/3}\frac{1}{\delta}\Big)}.
		\end{eqnarray}

		\vspace{\baselineskip}
		
		The above results are summarized as follows:
		
		\begin{eqnarray}
			\label{eq: T range}
			T &\leq& \frac{e\tau_{R,m}}{6N^2 b r} , \\
			\label{eq: g range}
			g &<& \min\Big(\frac{1}{2}, \frac{\tau_{R,m}}{N\tau_B(T,\mathfrak{J})}, \frac{\tau_{R,m}}{12N  {b}{r}{h}}, \frac{\tau_{R,m}^3}{2304 e^2  N^3{b}^2{r}^2{h}\tau_B(T,\mathfrak{J})}\Big) , \\
			\delta &<& \min{(\delta_1,\delta_2)},  \\
			\label{eq: sigma range}
			\sigma &>& \max{\Big(\frac{192\sqrt{2}  N  {b}{r}{h}^2}{\sqrt{\pi}\delta^2\tau_{R,m}}, \Big(\frac{\sqrt{\pi}}{4\sqrt{2}}\Big)^{1/3}\frac{1}{\delta}\Big)},
		\end{eqnarray}
		where
		\begin{eqnarray}
			\delta_1 &=& \Big( 12{b}{r}\Big(\frac{N}{\tau_{R,m}} + 2\Big(\frac{1}{g} + N\frac{\tau_B(T,\mathfrak{J})}{\tau_{R,m}}\Big){h} \Big)\Big)^{-1}, \\
			\delta_2 &=& \Big(12{b}{r}\Big(\frac{1}{\tau_{R,m}} + 2\Big(\frac{ \tau_{R,M}}{gN\tau_{R,m}} + \frac{{b} \tau_B(T,\mathfrak{J})}{\tau_{R,m}}\Big){h} \Big)\Big)^{-1},
		\end{eqnarray}
		and the time $t$ is decided in Eq. (\ref{eq: time step}). With Lemma \ref{lem:super} and these inequalities, Theorem \ref{the:complexity} is almost self-evident.

		\subsection{The iteration method of parameters}
		\label{app: iteration method}
		We now introduce the method for parameter iteration. Note that the functions limiting the parameter ranges (i.e., the right-hand side of inequalities (\ref{eq: g range}) - (\ref{eq: sigma range}) are polynomial functions of parameters $N,b,r,h,T$, one might consider iterating the parameters through simple halving (doubling). However, since the temperature $T$ also appears in these functions and needs to be adjusted during the iteration process to satisfy its corresponding inequality, the parameter update method requires careful consideration.
		
		Suppose the initial parameters: $T=T_0$, $g = g_0$, $\delta = \delta_0$, $\sigma = \sigma_0$. Set $g_0 < \min (1/2,\tau_R(T_0,J_{\text{S}})/\tau_B(T_0,J_{\text{S}}))$ to make sure that $\tau_B(T_0,J_{\text{S}}) < t_0 < \tau_R(T_0,J_{\text{S}})$ holds, where $t_0 = \tau_R(T_0,J_{\text{S}}) / g$. Given that the limiting functions for the ranges of these parameters (i.e., $g,\delta,\sigma,T$) vary polynomially with $T$, we can exponentially adjust the parameters step by step, nabling them to quickly fall within the appropriate range. The only difficulty is that $\tau_B(T,\mathfrak{J})$ may increase with $\mathrm{Poly}(1/T)$. In the following, we discuss the iterative approach for each parameter individually.
		
		The temperature $T$ is refreshed by $T \leftarrow T/2$ each time in the parameter loop, which is enough to make the inequality (\ref{eq: T range}) hold in $O(\log(N,b,r))$ times.
		
		The dimensionless rescaling factor $g$ is refreshed by $g \leftarrow  \lambda_1 g/2$, where
		
		\begin{eqnarray}
			\lambda_1 = \min\Big(1, \frac{\tau_B(T,\mathfrak{J})}{\tau_B(\frac{T}{2},\mathfrak{J})}\Big).
		\end{eqnarray}
		
		For inequality (\ref{eq: delta_1}), we define auxiliary function 
		
		\begin{eqnarray}
			f_1(g,T) = \frac{N}{\tau_{R,m}} + 2\Big(\frac{1}{g} + N\frac{\tau_B(T,\mathfrak{J})}{\tau_{R,m}}\Big){h}.
		\end{eqnarray}
		The inequality (\ref{eq: delta_1}) can be satisfied in $O(\log(N,b,r,h))$ iteration times if we refresh $\delta$ by $\lambda_2 \delta / 2$, where
		\begin{eqnarray}
			\lambda_2 = \min\Big(1,\frac{f_1(g,T)}{f_1(\frac{\lambda_1 g}{2},\frac{T}{2})}\Big).
		\end{eqnarray}

		Similarly, define auxiliary function 
		\begin{eqnarray}
			f_2(g,T) = \frac{1}{\tau_{R,m}} + 2\Big(\frac{\tau_{R,M}}{gN\tau_{R,m}} + \frac{{b} \tau_B(T,\mathfrak{J})}{\tau_{R,m}}\Big){h}, 
		\end{eqnarray}
		the inequality (\ref{eq: delta_2}) can be satisfied in $O(\log(N,b,r,h))$ iteration times if we refresh $\delta$ by $\lambda_3 \delta / 2$, where
		
		\begin{eqnarray}
			\lambda_3 = \min\Big(1,\frac{f_2(g,T)}{f_2(\frac{\lambda_1 g}{2},\frac{T}{2})}\Big)
		\end{eqnarray}
		
		Set $\lambda_4 = \min(\lambda_2,\lambda_3)$, the inequalities (\ref{eq: delta_1}) and (\ref{eq: delta_2}) can be satisfied simultaneously in $O(\log(N,b,r,h))$ iteration times if we refresh $\delta$ by $\lambda_4 \delta / 2$.
		
		To make sure that the inequality (\ref{eq: sigma range}) holds in $O(\log(N,b,r,h))$ iteration times, we refresh $\sigma$ by $2\lambda_5 \sigma$, where $\lambda_5 = 4/\lambda_4^2$.

		\section{An example of proper constraint function}
		\label{app: example}
		
		In this section, we present an proper constraint function and calculate a finite range of values for $\tau_R$ and $\tau_B$. Consider a constraint function that describes a spectral density with an exponential cutoff, that is 
		\begin{eqnarray}
			\label{eq: expcut J}
			\mathfrak{J(\omega)} = \frac{\omega^s}{\Omega^{s-1}} e^{-\omega/\Omega},
		\end{eqnarray}
		where $\Omega$ is the cutoff frequency and parameter $s$ determines the low-frequency behaviour of $\mathfrak{J(\omega)}$. For a spectral density of the form given by the Eq. (\ref{eq: expcut J}), one calls couplings with:
		\begin{eqnarray}
			&&s < 1: \qquad \text{sub-Ohmic} \notag \\
			&&s = 1: \qquad \text{Ohmic}   \notag \\
			&&s > 1: \qquad \text{sup-Ohmic}.
		\end{eqnarray}
		
		Calculating the correlation function $C(T,\mathfrak{J})$ in the form of Eq. (\ref{eq:correlationGB}) with $\mathfrak{J(\omega)}$ given in Eq. (\ref{eq: expcut J}) yields 
		\begin{eqnarray}
			C(T,\mathfrak{J}; t) = \frac{T^{s+1}\Gamma(s+1)}{\Omega^{s-1}}\Big[b\Big(s+1, \frac{1+\Omega/T-i\Omega t}{\Omega/T}\Big) + b\Big(s+1, \frac{1+i\Omega t}{\Omega/T}\Big)\Big],
		\end{eqnarray}
		where $\Gamma(s)$ is the Gamma function and 
		\begin{eqnarray}
			b(z,u) = \sum_{n=0}^\infty \frac{1}{(n+u)^z}, u \neq 0, -1, -2, \cdots
		\end{eqnarray}
		is the generalized Zeta function.
		
		Take $s = 3$, one can verify that the modulus of $C(T,\mathfrak{J}; t)$ satisfies
		\begin{eqnarray}
			\label{eq: example Ct}
			\abs{C(T,\mathfrak{J}; t)} &\leq& \frac{6T^4}{\Omega^2}\Big(\sum_{n=0}^\infty\frac{1}{\big(\big(n+\frac{1+\Omega/T }{\Omega/T}\big)^2 + t^2T^2\big)^2} + \sum_{n=0}^\infty \frac{1}{\big(\big(n+\frac{T}{\Omega}\big)^2 + t^2T^2\big)^2} \Big) \notag \\
			&\leq& \frac{6T^4}{\Omega^2} \Big(\frac{1}{\big(\frac{T^2}{\Omega^2} + t^2T^2\big)^2} + \sum_{n=1}^\infty \frac{2}{\big(n^2+t^2T^2\big)^2}\Big).
		\end{eqnarray}
		Thus, the integral of the modulus of $C(T,\mathfrak{J}; t)$ has an upper bound
		\begin{eqnarray}
			\int_0^\infty \abs{C(T,\mathfrak{J}; t)}dt \leq \frac{3\pi}{2}\Omega + \frac{6\pi T^3}{\Omega^2},
		\end{eqnarray}
		which means the inverse of $\tau_R(T,\mathfrak{J})$ has a finite upper bound.
		
		According to Eq. (\ref{eq:correlationGB}), we can calculate a lower bound of integral of the modulus of $C(T,\mathfrak{J}; t)$ by
		\begin{eqnarray}
			\int_0^\infty\abs{C(T,\mathfrak{J}; t)} dt \geq \int_0^\infty dt \abs{\int_0^\infty d\omega \mathfrak{J}(\omega) \sin{\omega t}} \geq 3\Omega,
		\end{eqnarray} 
		which means the inverse of $\tau_R(T,\mathfrak{J})$ has a finite lower bound. Together, $\tau_R(T,\mathfrak{J})$ is a finite value when temperature $T$ is finite.
		
		Furthermore, according to Eq. (\ref{eq: example Ct}), one can verify that the integral
		\begin{eqnarray}
			\int_0^\infty t\abs{C(T,\mathfrak{J}; t)} dt \leq 3 + \frac{\pi^2 T^2}{\Omega^2}
		\end{eqnarray}
		is finite. Thus, according to the definition of $\tau_B$ in Eq. (\ref{eq: tauB}), $\tau_B(T,\mathfrak{J})$ is also finite.
		
	\end{widetext}
	\bibliography{ref}
\end{document}